\documentclass[12pt]{article}

\usepackage[small]{caption}
\usepackage[doublespacing]{setspace}
\usepackage[noend]{algpseudocode}
\usepackage{floatrow}
\usepackage{xurl}
\usepackage[english]{babel}
\usepackage[utf8x]{inputenc}
\usepackage[T1]{fontenc}
\usepackage{adjustbox, algorithm, algpseudocode, amssymb, amsmath, amsthm, animate, apacite, bbm, bm, booktabs, color, diagbox, enumerate, epsfig, epsf, float, framed, graphicx, listings, lscape, mathrsfs, multicol, multirow, pstricks, pst-node, psfrag, rotating, siunitx, subcaption, threeparttable, tikz, verbatim, booktabs, multirow, mathtools, bbm, enumitem,dsfont}

\usetikzlibrary{arrows,calc,tikzmark}
\usepackage[colorlinks=true,linkcolor=blue,citecolor=blue]{hyperref}%
\usepackage{natbib}
\usepackage[toc,page]{appendix}
\makeatletter
\def\BState{\State\hskip-\ALG@thistlm}
\makeatother
\makeatletter
\newcounter{phase}[algorithm]
\newlength{\phaserulewidth}

\newcommand{\setphaserulewidth}{\setlength{\phaserulewidth}}
\newcommand*{\xdash}[1][1.5em]{\rule[0.5ex]{#1}{0.55pt}}

\makeatother
\setphaserulewidth{.7pt}

\newcounter{case}[algorithm]
\newlength{\caserulewidth}
\newcommand{\setcaserulewidth}{\setlength{\caserulewidth}}

\makeatother
\setcaserulewidth{.7pt}
\algdef{SE}[SUBALG]{Indent}{EndIndent}{}{\algorithmicend\ }%
\algtext*{Indent}
\algtext*{EndIndent}

\def\boxit#1{\vbox{\hrule\hbox{\vrule\kern6pt
			\vbox{\kern6pt#1\kern6pt}\kern6pt\vrule}\hrule}}

\def\bse{\begin{eqnarray*}}
	\def\ese{\end{eqnarray*}}
\def\be{\begin{eqnarray}}
\def\ee{\end{eqnarray}}
\def\bq{\begin{equation}}
\def\eq{\end{equation}}
\def\bse{\begin{eqnarray*}}
	\def\ese{\end{eqnarray*}}

\newcommand{\N}{\mathbb{N}}
\newcommand{\R}{\mathbb{R}}
\newcommand{\hl}{\hspace{0.6mm}}
\newcommand{\hs}{\hspace{0.3mm}}
\newcommand{\cid}{\xrightarrow[\hspace{5mm}]{d}}
\newcommand{\Cov}[2]{\mathrm{Cov}\{#1,\hspace{1mm} #2\}}

\theoremstyle{plain}

\newtheorem{property}{Proposition}
\newtheorem{lem}{Lemma}
\theoremstyle{definition}

\newtheorem{defi}{Definition}

\begin{document}

\thispagestyle{empty} \baselineskip=28pt \vskip 5mm
\begin{center} {\Huge{\bf Visualization and Assessment of Copula Symmetry}}
	
\end{center}

\baselineskip=12pt \vskip 10mm
\begin{center}\large
Cristian F. Jim\'enez-Var\'on\footnote[1]{\baselineskip=10pt Statistics Program,
King Abdullah University of Science and Technology,
Thuwal 23955-6900, Saudi Arabia.\\
E-mail: cristian.jimenezvaron@kaust.edu.sa, marc.genton@kaust.edu.sa, ying.sun@kaust.edu.sa
}, Hao Lee\footnote[2]{\baselineskip=10pt Statistics \& Data Science, Dietrich College of Humanities and Social Sciences, Carnegie Mellon University, United States.}, Marc G.~Genton\textcolor{blue}{$^1$}, and Ying Sun\textcolor{blue}{$^1$}
\end{center}

\baselineskip=17pt \vskip 10mm \centerline{\today} \vskip 15mm

\begin{center}
{\large{\bf Abstract}}
\end{center}

Visualization and assessment of copula structures are crucial for accurately understanding and modeling the dependencies in multivariate data analysis. In this paper, we introduce an innovative method that employs functional boxplots and rank-based testing procedures to evaluate copula symmetry. This approach is specifically designed to assess key characteristics such as reflection symmetry, radial symmetry, and joint symmetry. We first construct test functions for each specific property and then investigate the asymptotic properties of their empirical estimators. We demonstrate that the functional boxplot of these sample test functions serves as an informative visualization tool of a given copula structure, effectively measuring the departure from zero of the test function. Furthermore, we introduce a nonparametric testing procedure to assess the significance of deviations from symmetry, ensuring the accuracy and reliability of our visualization method. Through extensive simulation studies involving various copula models, we demonstrate the effectiveness of our testing approach. Finally, we apply our visualization and testing techniques to two real-world datasets: a nutritional habits survey with five variables and wind speed data from three locations in Saudi Arabia.

\baselineskip=14pt

\par\vfill\noindent

\clearpage\pagebreak\newpage \pagenumbering{arabic}
\baselineskip=26pt


\newpage
\section{Introduction}
Copula models have gained significant prominence in the field of statistics and data analysis due to their flexibility in modeling intricate dependence structures among random variables \citep{nelsen2006,joe2014,patton2012}. They have become indispensable tools for capturing and understanding various types of dependencies, such as tail dependence, asymmetry, and nonlinearity \citep{genest2007,cherubini2004}. By decoupling the marginal distributions from the dependence structure, copula models offer a powerful framework for accurately characterizing complex multivariate relationships \citep{cherubini2004,joe2014}. In addition to their theoretical significance, copula models have found extensive real-world applications in finance, insurance, and environmental sciences. In finance, for instance, copula models facilitate portfolio optimization, risk management, and pricing of complex financial derivatives by accurately modeling dependencies between financial assets \citep{cherubini2004,patton2012}. 

According to the representation theorem provided by \cite{1959.S.PISUP}, every multivariate cumulative distribution function, $F$, of a continuous random vector $\mathbf{X}=(X_1,\ldots, X_d)^\top$ on $\R^d$, can be written as
\begin{equation}
        F(x_1,\ldots,x_d)=\mathds{P}(X_1\leq x_1,\ldots,X_d \leq x_d)= C\{F_1(x_1),\ldots,F_d(x_d)\},
\label{FX}
\end{equation}
where $F_l(x_l)=\mathds{P}(X_l\leq x_l)$, $x_l\in \R$, are the continuous marginal distributions, $l=1,\ldots,d$, and $C:[0,1]^d \rightarrow [0,1]$ is the unique copula that characterizes the dependence structure of the random vector $\mathbf{X}$ and can be obtained from
\begin{equation}
        C(u_1,\ldots,u_d)=\mathds{P}(U_1\leq u_1,\ldots,U_d \leq u_d)= F\{F_1^{-1}(u_1),\ldots, F_d^{-1}(u_d)\},
    \label{Cu}
\end{equation}
where $F_l^{-1}(u_l)=\text{inf}\{x|F_l(x)\geq u_l\}$, $u_l\in [0,1]$, is the quantile function of $F_l$ and $\mathbf{U}=(U_1,\ldots,U_d)^\top$ with $U_l=F_l(X_l)$.  From \eqref{Cu}, the copula function $C$ serves as a cumulative distribution function for the random vector $\mathbf{U}$, residing within the $d$-dimensional unit hypercube and characterized by its marginal distributions. In practical applications, the representation provided in Equation \eqref{FX} enables the modeling of the dependence structure, given the knowledge of the marginal distributions. This can be achieved by selecting an appropriate parametric copula model from a wide range of options available in the literature \citep[see e.g.,][]{nelsen2006,joe2014}.

Choosing an appropriate copula model is a challenging task when quantifying dependence. In various practical applications, such as actuarial science, finance, and survival analysis \citep{nelsen2006,patton2006,Aas2009}, the common approach has been to rely on expert knowledge or choose a copula model based on mathematical convenience rather than its suitability for the specific data application. However, this approach can introduce limitations and biases in the analysis \citep{Mikosch2006,nelsen2006,Aas2009,joe2014}.

Several existing approaches in the literature for copula model selection are based on goodness-of-fit tests for copulas \citep{genest2007,Genest2009, Berg2009}. These methods typically treat the univariate marginal distribution as an infinite-dimensional nuisance parameter and replace the observations with maximally invariant statistics, such as ranks. 

Understanding the properties and structure of copulas is crucial for capturing and interpreting the relationships between random variables. Various methods have been proposed in the literature to specify and test copula structures.  \cite{Jaworski2010} introduced a test for the associativity structure of copulas based on the asymptotic distribution of the pointwise copula estimator. However, this test only assesses associativity at a specific point rather than for the entire copula process, as discussed in \cite{BUCHER2012}. \cite{BUCHER2012} derived Cram\'er-von Mises and Kolmogorov-Smirnov type test statistics for evaluating the characteristics of associativity. Additionally, they developed test statistics for Archimedean copulas \citep{BUCHER2012}. \cite{Bucher2011} proposed a test for extreme value dependence based on the minimum weighted $L^2$-distance of extreme-value copulas. The bivariate symmetry test for copulas, based on Cram\'er-von Mises and Kolmogorov-Smirnov functionals of the rank-based empirical copula process, was introduced by \cite{2012.GNQ.AISM} and \cite{2014.GN.SP}.

\cite{2013.LG.JASA} proposed a nonparametric method for identifying copula symmetry using the asymptotic distribution of the empirical copula process. \cite{Quessy2016} developed a statistical framework based on quadratic functionals to test the identity of copulas from a multivariate distribution.
More recently, \cite{Jaser2021} proposed simpler nonparametric tests for the symmetry and radial symmetry of bivariate copulas. Their approach involves creating two bivariate samples by manipulating the underlying copula while preserving its dependence structure. The test statistics are based on the difference between the empirical Kendall's tau of both samples.

In this paper, we present a new approach for visualizing and testing the structure of copula models, specifically focusing on properties such as symmetry, radial symmetry, and joint symmetry as defined in \cite{1993.N.JNS}. Our approach complements existing goodness-of-fit tests for copula model selection. To visualize these copula structures, we employ the functional boxplot introduced by \cite{2011.SG.JCGS} as a visual tool to quantify the deviations from a given copula structure by measuring the departure from zero of sample test functions. We demonstrate that these visualizations offer insights into the extent to which specific copula structures are adhered to.

Additionally, we introduce a nonparametric testing procedure to assess the significance of deviations from symmetry. This testing procedure is motivated by the techniques proposed by \cite{2019.HS.SS} and \cite{2023.HSG.JCGS}, which utilize a functional data framework to visualize and assess spatio-temporal covariance properties in both univariate and multivariate cases. We evaluate the effectiveness of our proposed testing approach through extensive simulation studies involving various copula models.

The paper is structured as follows. In Section~\ref{Methods}, we outline the copula symmetry of interest, including reflection symmetry, radial symmetry, and joint symmetry. We also provide details on the visualization and nonparametric testing procedures for each of these copula structures. Section~\ref{Simulation} presents the simulation results regarding the size and power of our proposed nonparametric test. In Section~\ref{data_app}, we apply our methods to two real-world datasets: a nutritional habits survey with five variables and wind speed data from three locations in Saudi Arabia. Finally, the paper concludes with a discussion in Section~\ref{Discuss}.

\section{Methodology}\label{Methods}

In Section~\ref{cop_sym}, we introduce copula symmetries. Section~\ref{sec:test_f} covers the construction of test functions and provides asymptotic results for proper estimators. We present the visualization of test functions with functional boxplots in Section~\ref{Vis:FB}. Lastly, in Section~\ref{rank-test}, we describe a rank-based testing procedure for copula symmetry.

\subsection{Copula Symmetry}\label{cop_sym}
Our discussion centers on the symmetry of bivariate copulas. Unlike the case of univariate functions, the concept of symmetry is not uniquely defined in a multivariate setting. Therefore, different notions of symmetry have been investigated in the context of copulas. Here we focus on the ones presented in~\cite{1993.N.JNS}. 
\begin{defi} \label{def:sym}
A copula $C$ is said to be \textit{symmetric} if 
\begin{equation} \label{eq:sym}
		C (u, v) - C(v, u) = 0, \hspace{5mm} \forall (u, v) \in [0, 1]^2.
\end{equation}
\end{defi}
Based on the algebraic Equation~\eqref{eq:sym}, one should notice that for any symmetric copula $C$, its distribution is symmetric with respect to the diagonal connecting the origin and the point $(1, 1)$. Thus, the symmetry in Definition~\ref{def:sym} is called \textit{reflection symmetry} in some literature. We will also use reflection symmetry to refer to this type of symmetry in the sequel. 

\begin{defi} \label{def:rsym}
A copula $C$ is said to be \textit{radially symmetric} if 
\begin{equation} \label{eq:rsym}
		C (u, v) - C(1 - u, 1 - v) + 1 - u - v = 0, \hspace{5mm} \forall (u, v) \in [0, 1]^2.
	\end{equation}
\end{defi}
Equivalently, one can state that the radial symmetry property as $C(u, v) - C^\ast(u, v) = 0$ for all $(u, v) \in [0, 1]^2$, where $C^\ast$ stands for the survival copula associated with $C$, i.e., for all $(u, v) \in [0, 1]^2$, $C^\ast (u, v) = C(1 - u, 1 - v) - 1 + u + v$. Pointed out by~\cite{1993.N.JNS}, there exist copulas that are reflection symmetric
but not radially symmetric and, conversely, copulas that are
radially symmetric but not reflection symmetric.
\begin{defi} \label{def:jsym}
A copula $C$ is said to be \textit{jointly symmetric} if it satisfies  
\begin{equation*}
		C (u, v) + C(u, 1 - v) - u = 0 
	\end{equation*}
	and 
	\begin{equation*}
		C (u, v) + C(1- u, v) - v = 0
	\end{equation*}
	for all $ (u, v) \in [0, 1]^2$. 
\end{defi}
One can easily show that joint symmetry implies radial symmetry. However, there is no implication between reflection symmetry and joint symmetry. The Figure 1 of \cite{2013.LG.JASA} provides the interrelations among the three types of symmetry. 

\subsection{Test functions} \label{sec:test_f}

We propose to assess and visualize the three types of copula symmetries by the construction of test functions. First, we focus on the construction of test functions specifically for reflection symmetry (S). For any fixed $v \in [0, 1]$, define the reflection symmetry test functions $f^{^{S}}_{v}, g^{^{\hspace{-0.3mm}S}}_{v}: [0,1] \rightarrow [-1, 1]$ by 
\begin{equation*}
    f^{^{S}}_{v} (t) = C(t, v) - C(v, t) \text{\hspace{5mm}and\hspace{5mm}} g^{^{\hspace{-0.3mm}S}}_{v} (t) = - f^{^{S}}_{v} (t)
\end{equation*}
for all $t \in [0, 1]$. If $C$ is reflection symmetric, we have $f^{^{S}}_{v} = g^{^{\hspace{-0.3mm}S}}_{v} = 0$; otherwise, the values of $f^{^{S}}_{v}$ and $g^{^{\hspace{-0.3mm}S}}_{v}$ vary with respect to $t$. For all $m \in \N$, any set $\{v_1, \ldots, v_m, w_1, \ldots, w_m\} \subseteq [0, 1]$ can, in the same manner, induce $2m$ reflection symmetry test functions $f^{^{S}}_{v_1}, \ldots, f^{^{S}}_{v_m}$ and $ g^{^{\hspace{-0.3mm}S}}_{w_1}, \ldots, g^{^{\hspace{-0.3mm}S}}_{w_m}$. 

In practice, we need proper estimators $\widehat{f}^{^{\hspace{0.4mm}S}}_{v_1}, \ldots, \widehat{f}^{^{\hspace{0.4mm}S}}_{v_m}$ and $\widehat{g}^{^{\hspace{0.5mm}S}}_{w_1}, \ldots, \widehat{g}^{^{\hspace{0.4mm}S}}_{w_m}$ of the test functions $f^{^{S}}_{v_1}, \ldots, f^{^{S}}_{v_m}$ and $ g^{^{\hspace{-0.3mm}S}}_{w_1}, \ldots, g^{^{\hspace{-0.3mm}S}}_{w_m}$. Intuitively, the estimators can be obtained through a linear combination of the corresponding empirical copulas. To provide a clearer definition of these estimators, we briefly summarize the important asymptotic results of the two-dimensional empirical copula process \citep[see, e.g.,][]{1979.D.CRASP,1984.S.AP,2005.T.CJS}. Additional recent results on convergence rate can be found in \cite{2010.GS.JMA}, \cite{2012.S.Bernoulli}, \cite{2013.SA.SPL}, and the references therein.

Let $C$ be a $2$-dimensional copula and $\bm{U} = (U_1, U_2)^\top$ be a random vector that is $C$-distributed. If one has direct access to a random sample $\bm{U}_1, \ldots, \bm{U}_n$ of size $n$, the copula $C$ can be estimated by the empirical copula process defined by 
\begin{equation*}
    \widehat{C}_n(u, v) = \frac{1}{n} \sum_{i = 1}^n \mathbbm{1} \big \{U_{i1} \le u, U_{i2} \le v \big \}, \hspace{5mm} \forall (u, v) \in [0,1]^2, 
\end{equation*}

\noindent where $\mathbbm{1}(\cdot)$ denotes the indicator function. As a well-known result, we have that for all $(u, v) \in [0, 1]^2$,
\begin{equation} \label{eq:ec1}
    \mathbb{C}_n (u, v) := \sqrt{n} \left \{ \widehat{C}_n(u, v) - C(u, v) \right \} \cid \mathbb{C} (u, v) \text{\hspace{2mm}in\hspace{2mm}} \ell^\infty ([0 ,1]^2), \hspace{5mm} n \rightarrow \infty, 
\end{equation} 
 
\noindent where $\ell^\infty ([0, 1]^2)$ denotes the space of all the bounded functions over the compact set $[0, 1]^2$ and $\mathbb{C}$ is a $2$-dimensional pinned $C$-Brownian sheet, i.e., it is a centered Gaussian random field with the covariance function given by
\begin{equation*}
\Cov{\mathbb{C} (u, v)}{\mathbb{C}(x, y)} = C(u \land x, v \land y ) - C(u, v) C(x, y), \hspace{5mm} \forall u, v, x, y \in [0, 1]
\end{equation*}
where for all $a, b \in \R$, $a \land b = \min\{a, b\}$ \citep{2012.GNQ.AISM}. However, it is often the case that the $n$ observations we have are generated from a random vector $\bm{X} = (X_1, X_2)^\top$ that is not necessarily uniformly distributed over the interval $[0,1]$. The representation theorem \citep{1959.S.PISUP} states that it can be expressed as the composition of a copula $C$ and marginals of $\bm{X}$. In this case, from every $\bm{X}_i$ in the random sample, one can estimate a  $\bm{U}_i$ by the pseudo-observation $\widehat{\bm{U}}_i = (\widehat{U}_{i1}, \widehat{U}_{i2})^\top$, where for all $s \in \{1, 2\}$, 
\begin{equation*}
    \widehat{U}_{is} = \frac{1}{n} \sum_{r = 1}^n \mathbbm{1} \big \{ X_{rs} \le X_{is} \big \}.
\end{equation*}
With all of these $\widehat{\bm{U}}_i$'s, one can estimate $C$ by
\begin{equation*}
    \widehat{D}_n(u, v) = \frac{1}{n} \sum_{i = 1}^n \mathbbm{1} \big \{\widehat{U}_{i1} \le u, \widehat{U}_{i2} \le v \big \}, \hspace{5mm} \forall (u, v) \in [0,1]^2. 
\end{equation*}

It has been shown that when $C$ is \textit{regular} (see Definition 1 in \cite{2012.GNQ.AISM} for example), or loosely speaking, when $C$ is differentiable with continuous partials, we have
\begin{multline} \label{eq:ec2}
    \widehat{\mathbb{D}}_n (u, v) := \sqrt{n} \Big \{\widehat{D}_n (u, v) - C(u, v) \Big \} \\
     \cid \mathbb{C}(u, v) - \dot{C}_1 (u, v) \mathbb{C} (u, 1) - \dot{C}_2 (u, v) \mathbb{C} (1, v) \text{\hspace{2mm}in\hspace{2mm}} \ell^\infty ([0 ,1]^2), \hspace{5mm} n \rightarrow \infty,
\end{multline}
where $\dot{C}_1$ and $\dot{C}_2$ denote the partial derivatives of the copula $C$ with respect to its first and second variables, respectively. In the sequel, if not otherwise stated, we always impose the regularity assumption on the underlying copula $C$. To sum up, the estimator of the reflection symmetry test function $\widehat{f}^{^{\hspace{0.5mm}S}}_{v}$ can be defined by either 
\begin{equation*}
    \widehat{f}^{^{\hspace{0.5mm}S}}_{v} (t) = \widehat{C}_n (t, v) - \widehat{C}_n (v, t), \hspace{5mm} \forall t \in [0,1],
\end{equation*}
or 
\begin{equation*}
    \widehat{f}^{^{\hspace{0.5mm}S}}_{v} (t) = \widehat{D}_n (t, v) - \widehat{D}_n (v, t), \hspace{5mm} \forall t \in [0,1], 
\end{equation*}
depending on the type of the data set we have, and we simply set $\widehat{g}^{^{\hspace{0.4mm}S}}_v(t) = -\widehat{f}^{^{\hspace{0.5mm}S}}_{v}(t)$. 

\begin{property}\label{prop:a1}
Given any fixed $v \in [0, 1]$, the estimator $\widehat{f}^{^{\hspace{0.5mm}S}}_{v}(t)$ satisfies that for all $t \in [0, 1]$, 
\begin{equation*}
    \sqrt{n} \left \{ \widehat{f}^{^{\hspace{0.5mm}S}}_{v}(t) -  f^{^{S}}_{v} (t) \right \} \cid 
    \begin{dcases}
        \mathbb{E}^{^{S}}_v (t) & \text{if } \widehat{f}^{^{\hspace{0.5mm}S}}_{v} (t) = \widehat{C}_n (t, v) - \widehat{C}_n (v, t)\\
        \widehat{\mathbb{E}}^{^S}_v (t) & \text{if } \widehat{f}^{^{\hspace{0.5mm}S}}_{v} (t) = \widehat{D}_n (t, v) - \widehat{D}_n (v, t)\\
    \end{dcases}
\end{equation*}
in $\ell^\infty([0, 1])$ as $n \rightarrow \infty$, where $\mathbb{E}^{^{S}}_v$ and $\widehat{\mathbb{E}}^{^S}_v$ are two centered Gaussian random fields. 
\end{property}

\noindent The proof of this Proposition can be found in Appendix~\ref{A:prop1}.

\begin{table}[htb!]
\centering
    \begin{tabular}{c|c|c}
    \hline
    \\[-0.8em]
         \textbf{Type} & \textbf{Access} & \textbf{Test Functions} \\
         \\[-0.8em]
         \hline
          \\[-0.8em]
         
         \multirow{3}{*}{S} & $(U_1, U_2)$ &
         $\widehat{f}^{^{\hspace{0.5mm}S}}_{v} (t) = \widehat{C}_n (t, v) - \widehat{C}_n (v, t), \hspace{5mm} \widehat{g}^{^{\hspace{0.4mm}S}}_v (t) = - \widehat{f}^{^{\hspace{0.5mm}S}}_{v} (t)$ \\[1ex]
        \cmidrule{2-3}
        
         & $(X_1, X_2)$ & $\widehat{f}^{^{\hspace{0.5mm}S}}_{v} (t) = \widehat{D}_n (t, v) - \widehat{D}_n (v, t)\hspace{5mm} \widehat{g}^{^{\hspace{0.4mm}S}}_v (t) = - \widehat{f}^{^{\hspace{0.5mm}S}}_{v} (t)$ \\[1ex]
         \hline
         \\[-0.8em]
         
         \multirow{3}{*}{R} & $(U_1, U_2)$ & $\widehat{f}^{^{\hspace{0.5mm}R}}_{v} (t) = \widehat{C}_n (t, v) - \widehat{C}_n (1-t, 1-v) + 1 - t - v, \hspace{5mm} \widehat{g}^{^{\hspace{0.4mm}R}}_v (t) = - \widehat{f}^{^{\hspace{0.5mm}R}}_{v} (t)$ \\[1ex]
        \cmidrule{2-3}
        
        & $(X_1, X_2)$ &  $\widehat{f}^{^{\hspace{0.5mm}R}}_{v} (t) = \widehat{D}_n (t, v) - \widehat{D}_n (1-t, 1-v) + 1 - t - v, \hspace{5mm} \widehat{g}^{^{\hspace{0.4mm}R}}_v (t) = - \widehat{f}^{^{\hspace{0.5mm}R}}_{v} (t)$\\[1ex]
         \hline
         \\[-0.8em]
         
        \multirow{7}{*}{J} & \multirow{3}{*}{$(U_1, U_2)$} & $\widehat{f}^{^{\hspace{0.5mm}J, 1}}_{v} (t) = \widehat{C}_n (t, v) + \widehat{C}_n (t, 1 - v) - t, \hspace{5mm} \widehat{g}^{^{\hspace{0.4mm}J, 1}}_v (t) = - \widehat{f}^{^{\hspace{0.5mm}J, 1}}_{v} (t)$ \\[1ex]
         \cmidrule{3-3}
         
        & & $\widehat{f}^{^{\hspace{0.5mm}J, 2}}_{v} (t) = \widehat{C}_n (t, v) + \widehat{C}_n (1 - t, v) - v, \hspace{5mm} \widehat{g}^{^{\hspace{0.4mm}J, 2}}_v (t) = - \widehat{f}^{^{\hspace{0.5mm}J, 2}}_{v} (t)$ \\[1ex]
         \cmidrule{2-3}
         
        & \multirow{3}{*}{$(X_1, X_2)$} & $\widehat{f}^{^{\hspace{0.5mm}J, 1}}_{v} (t) = \widehat{D}_n (t, v) + \widehat{D}_n (t, 1 - v) - t, \hspace{5mm} \widehat{g}^{^{\hspace{0.4mm}J, 1}}_v (t) = - \widehat{f}^{^{\hspace{0.5mm}J, 1}}_{v} (t)$ \\[1ex]
         \cmidrule{3-3}
         
         & & $\widehat{f}^{^{\hspace{0.5mm}J, 2}}_{v} (t) = \widehat{D}_n (t, v) + \widehat{D}_n (t, 1 - v) - t, \hspace{5mm} \widehat{g}^{^{\hspace{0.4mm}J, 2}}_v (t) = - \widehat{f}^{^{\hspace{0.5mm}J, 2}}_{v} (t)$ \\[1ex]
         \hline
    \end{tabular}
    \caption{Definitions of the estimators of the test functions for different types of symmetry.}
    \label{tb:tf}
\end{table}

Similarly, we can construct radial symmetry (R) test functions, denoted as $\widehat{f}^{^{\hspace{0.5mm}R}}_{v}(t)$ and $\widehat{g}^{^{\hspace{0.4mm}R}}_{v}(t)$, based on Definition~\ref{def:rsym} to estimate $f^{^{R}}_v(t)$ and $g^{^{\hspace{-0.3mm}R}}_v(t)$, as described in Table~\ref{tb:tf}. The convergence of these estimators is presented in Proposition~\ref{prop:a2}.

\begin{property} \label{prop:a2}
Given any fixed $v \in [0, 1]$, the estimator $\widehat{f}^{^{\hspace{0.5mm}R}}_{v}(t)$ satisfies that for all $t \in [0, 1]$, 
\begin{equation*}
    \sqrt{n} \left \{ \widehat{f}^{^{\hspace{0.5mm}R}}_{v}(t) -  f^{^{R}}_{v} (t) \right \} \cid 
    \begin{dcases}
        \mathbb{E}^{^R}_v (t) & \text{if } \widehat{f}^{^{\hspace{0.5mm}R}}_{v} (t) = \widehat{C}_n (t, v) - \widehat{C}_n (1-t, 1-v) + 1 - t - v\\
        \widehat{\mathbb{E}}^{^R}_v (t) & \text{if } \widehat{f}^{^{\hspace{0.5mm}R}}_{v} (t) = \widehat{D}_n (t, v) - \widehat{D}_n (1-t, 1-v) + 1 - t - v
    \end{dcases}
\end{equation*}
in $\ell^\infty ([0, 1])$ as $n \rightarrow \infty$, where $\mathbb{E}^{^R}_v$ and $\widehat{\mathbb{E}}^{^R}_v$ are two centered Gaussian random fields. 
\end{property}

\noindent The proof of this Proposition can be found in Appendix~\ref{A:prop2}.

As for the test functions for joint symmetry (J), one should investigate the two properties in Definition~\ref{def:jsym} separately. Hence, for any given $v \in [0,1]$, we construct four population joint symmetry test functions: $f^{^{J, 1}}_v, f^{^{J, 2}}_v, g^{^{\hspace{-0.3mm}J, 1}}_v, g^{^{\hspace{-0.3mm}J, 2}}_v : [0, 1] \rightarrow [-1, 1]$ that are given by 
\begin{equation*}
    f^{^{J, 1}}_v (t) = C(t, v) + C(t, 1 - v) - t \text{\hspace{5mm} and \hspace{5mm}} g^{^{\hspace{-0.3mm}J, 1}}_v (t) = - f^{^{J, 1}}_v (t)
\end{equation*}
as well as 
\begin{equation*}
    f^{^{J, 2}}_v (t) = C( 1- t, v) + C(t, 1 - v) - v \text{\hspace{5mm} and \hspace{5mm}} g^{^{\hspace{-0.3mm}J, 2}}_v (t) = - f^{^{J, 2}}_v (t)
\end{equation*}
for all $t \in [0, 1]$. As the cases of other symmetries, they can be estimated by the test functions involving empirical copulas (see Table~\ref{tb:tf} for their definitions), whose relevant asymptotic results are presented in Proposition~\ref{prop:a3}.

\begin{property} \label{prop:a3}
Given $v \in [0, 1]$, the estimator $\widehat{f}^{^{\hspace{0.5mm}J, 1}}_{v}(t)$ satisfies that for all $t \in [0, 1]$, 
\begin{equation*}
    \sqrt{n} \left \{ \widehat{f}^{^{\hspace{0.5mm}J, 1}}_{v}(t) -  f^{^{J, 1}}_{v} (t) \right \} \cid 
    \begin{dcases}
        \mathbb{E}^{^{J, 1}}_v (t) & \text{if } \widehat{f}^{^{\hspace{0.5mm}J, 1}}_{v} (t) = \widehat{C}_n (t, v) +  \widehat{C}_n (t, 1-v) - t \\
        \widehat{\mathbb{E}}^{^{J, 1}}_v (t) & \text{if } \widehat{f}^{^{\hspace{0.5mm}J, 1}}_{v} (t) = \widehat{D}_n (t, v) +  \widehat{D}_n (t, 1-v) - t
    \end{dcases}
\end{equation*}
in $\ell^\infty ([0, 1])$ as $n \rightarrow \infty$, where $\mathbb{E}^{^{J, 1}}_v$ and $\widehat{\mathbb{E}}^{^{J, 1}}_v$ are two centered Gaussian random fields. 

\noindent Similarly, given $v \in [0, 1]$, the estimator $\widehat{f}^{^{\hspace{0.5mm}J, 2}}_v(t)$ satisfies that for all $t \in [0, 1]$, 
\begin{equation*}
    \sqrt{n} \left \{ \widehat{f}^{^{\hspace{0.5mm}J, 2}}_{v}(t) -  f^{^{J, 2}}_{v} (t) \right \} \cid 
    \begin{dcases}
        \mathbb{E}^{^{J, 2}}_v (t) & \text{if } \widehat{f}^{^{\hspace{0.5mm}J, 2}}_{v} (t) = \widehat{C}_n (t, v) +  \widehat{C}_n (1 - v, t) - v \\
        \widehat{\mathbb{E}}^{^{J, 2}}_v (t) & \text{if } \widehat{f}^{^{\hspace{0.5mm}J, 2}}_{v} (t) = \widehat{D}_n (t, v) +  \widehat{D}_n (1 -v, t) - v
    \end{dcases}
\end{equation*}
in $\ell^\infty ([0, 1])$ as $n \rightarrow \infty$, where $\mathbb{E}^{^{J, 2}}_v$ and $\widehat{\mathbb{E}}^{^{J, 2}}_v$ are two centered Gaussian random fields. 
\end{property}

\noindent The proof of this Proposition can be found in Appendix~\ref{A:prop3}.

\subsection{Visualization}\label{Vis:FB}

The visualization of test functions can be achieved using functional boxplots \citep{2011.SG.JCGS}. To construct the functional boxplot of $\widehat{f}^{^{\hspace{0.5mm}S}}_{v_1}, \ldots, \widehat{f}^{^{\hspace{0.5mm}S}}_{v_m}, \widehat{g}^{^{\hspace{0.5mm}S}}_{w_1}, \ldots, \widehat{g}^{^{\hspace{0.5mm}S}}_{w_m}$, we need to discretize the interval $[0, 1]$ into an evenly spaced range of $p$ points for some properly chosen $p \in \N$ and then evaluate the test functions $\widehat{f}^{^{\hspace{0.5mm}S}}_{v_1}, \ldots, \widehat{f}^{^{\hspace{0.5mm}S}}_{v_m}, \widehat{g}^{^{\hspace{0.4mm}S}}_{w_1}, \ldots, \widehat{g}^{^{\hspace{0.4mm}S}}_{w_m}$ over them. We denote the range of $p$ points by $\{0 \leq t_1, \ldots, t_p \leq 1\}$. According to Proposition~\ref{prop:a1}, for every $j \in \{1, \ldots, m\}$, the joint distribution of the random vector $\big ( \widehat{f}^{^{\hspace{0.5mm}S}}_{v_j} (t_1) - f^{^{S}}_{v_j} (t_1), \ldots, \widehat{f}^{^{\hspace{0.5mm}S}}_{v_j} (t_p) - f^{^{S}}_{v_j} (t_p) \big )^\top$ is close to a multivariate normal distribution with mean zero when the sample size $n$ is sufficiently large, and so is the joint distribution of the random vector $\big ( \widehat{g}^{^{\hspace{0.4mm}S}}_{w_j} (t_1) - g^{^{\hspace{-0.3mm}S}}_{w_j} (t_1), \ldots, \widehat{g}^{^{\hspace{0.4mm}S}}_{w_j} (t_p) - g^{^{\hspace{-0.3mm}S}}_{w_j} (t_p) \big )^\top$. Thus, the functional boxplot should be fairly centered around zero (even though its shape should not be expected as a horizontal band since the variance at each point $t_k, k \in \{1, \ldots, p\},$ can vary based on the underlying copula $C$). If the underlying copula $C$ is reflection symmetric, it implies the equalities $f^{^{S}}_{v_1} \equiv \cdots \equiv f^{^{S}}_{v_m} \equiv g^{^{\hspace{-0.3mm}S}}_{w_1} \equiv \cdots \equiv g^{^{\hspace{-0.3mm}S}}_{w_m} \equiv 0$. Otherwise, the functional boxplot should be of an irregular shape that is visually deviated from zero. In the special case when $v_i \approx w_i$ for a great proportion of $i \in \{1, \ldots, n\}$, the functional boxplot may form an irregular envelope that is almost symmetric with respect to zero. 

To assess radial symmetry, we analyze the functional boxplots of $\widehat{f}^{^{\hspace{0.5mm}R}}_{v_1}, \ldots, \widehat{f}^{^{\hspace{0.5mm}R}}_{v_m}, \widehat{g}^{^{\hspace{0.4mm}R}}_{w_1}, \ldots, \widehat{g}^{^{\hspace{0.4mm}R}}_{w_m}$. According to Proposition~\ref{prop:a2}, the functional boxplots should be fairly centered around zero if and only if the underlying copula $C$ is indeed radial symmetric. For joint symmetry, two types of test functions are defined, resulting in two functional boxplots per copula. Proposition~\ref{prop:a3} states that both functional boxplots should be centered around zero if the underlying copula is jointly symmetric. 

To stress the interpretation, we demonstrate the functional boxplots based on various copulas models in Figure~\ref{fig:vs2}. The gradation of the colors used in the plots manifests the density of functional data: the darker the colors are, the more functional curves are located in place. The corresponding structure(s) of each copula is flagged by the abbreviation(s) in bold: \textbf{S} (reflection symmetry), \textbf{R} (radial symmetry), \textbf{J} (joint symmetry) and \textbf{AS} (asymmetry) in the title of their respective functional boxplots. 

\begin{figure}[b!]
    \centering
    \includegraphics[width=\textwidth, height=15cm]{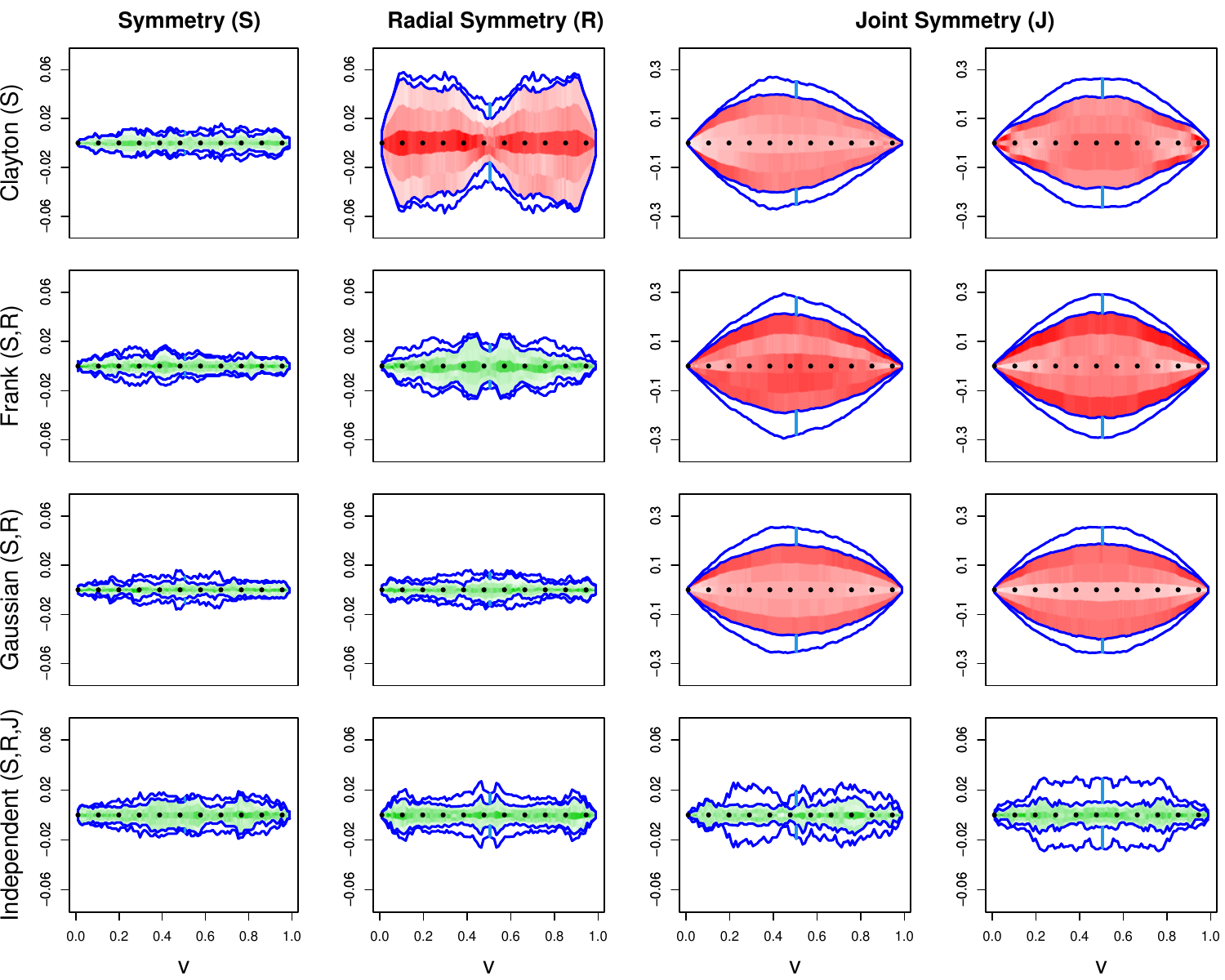}
    \label{fig:vs1}
\end{figure}

\begin{figure}[ht!]
    \centering
    \includegraphics[width=\textwidth, height=15cm]{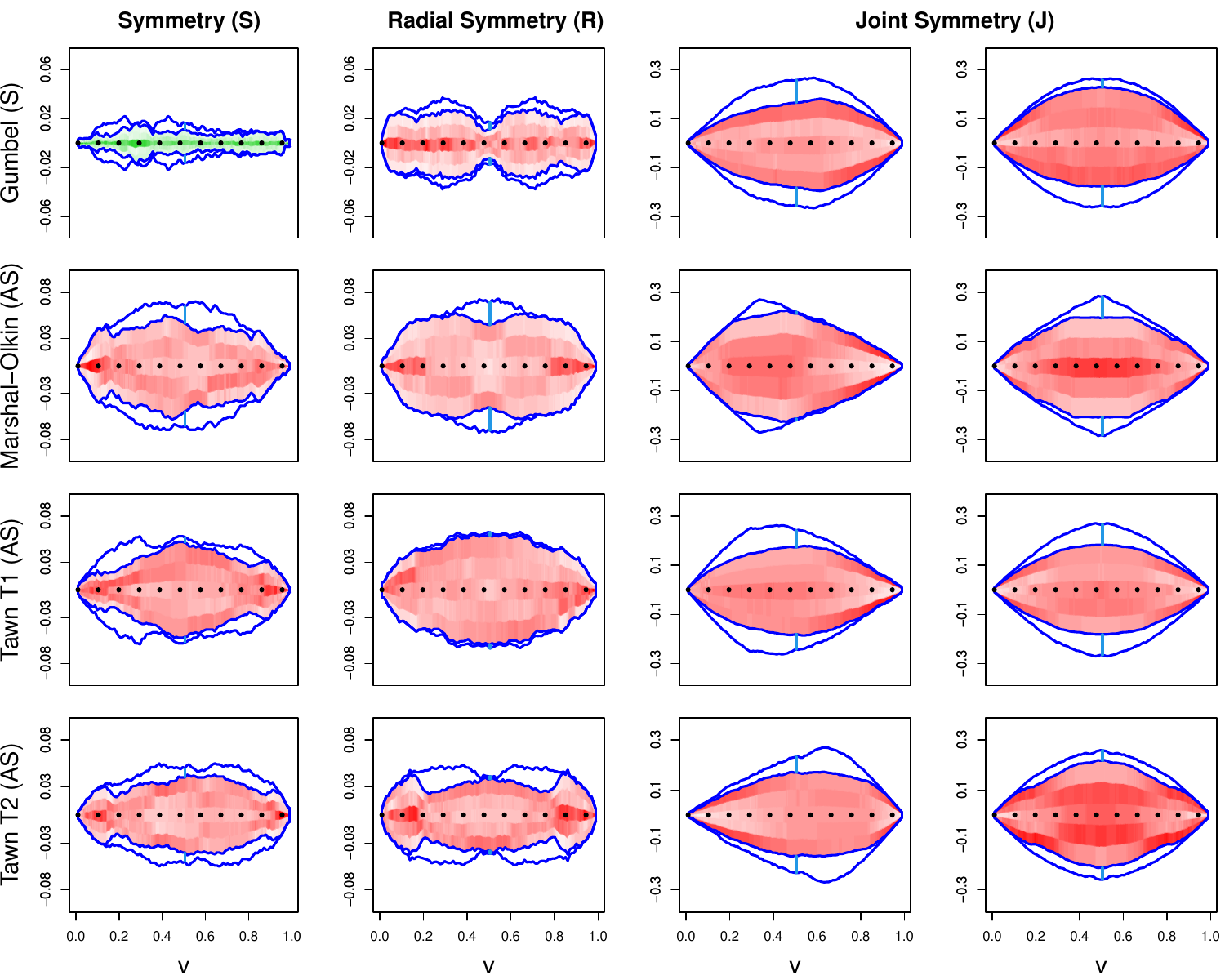}
    \caption{\textbf{Functional boxplots for the visualization of copula symmetries}. The parameter(s) of each copula is chosen to make its Kendall's tau equal/close to $0.5$. Especially, the Marshall-Olkin copula takes the parameters $(0.55, 0.85)$, and the Tawn copula, regardless of its types, takes the parameters $(4.28, 0.60)$. The corresponding test functions are constructed with $n = 1000, m = 250$ and $p = 100$, and the values of $v_j, w_j$ for all $j \in \{1, \ldots, m\}$ are randomly chosen from the unit interval $[0,1]$. Note that the plots of copulas that inherit the specified symmetries are featured by green-centered regions, whereas those who disobey the structures are featured by red ones.}
    \label{fig:vs2}
\end{figure}

\subsection{Rank-based hypothesis testing procedure}\label{rank-test}

In Section~\ref{Vis:FB}, we explored the visualization of copula structures present in a given sample or dataset. We discussed the relevant asymptotic results and demonstrated how functional boxplots of selected test functions can provide intuitive indications of various symmetries. To thoroughly investigate and quantify these symmetries, we introduce a one-sample ranked-based hypothesis testing procedure. This procedure is a modification of the methods proposed by \cite{2019.HS.SS} and \cite{2023.HSG.JCGS} for testing the separability and symmetry of univariate and multivariate covariance functions. Our adapted approach allows us to assess the presence of symmetries in copula structures in a robust and statistically rigorous manner. Both of these methods can be considered as adaptations of the two-sample rank-based test proposed by \cite{2009.LR.JASA}. The original test is designed to determine whether two sets of functional data are derived from the same distribution. In a similar manner, we modify this test to examine reflection symmetry as the initial step. Subsequently, by making specific adjustments, the procedure can be extended to test radial and joint symmetry.

Suppose that we have the observations $\bm{U}_i$'s or the pseudo-observations $\widehat{\bm{U}}_i$'s for $i = 1, \ldots, n$. As expected, for $v \in [0,1]$ the null $\mathcal{H}_0$ and the alternative $\mathcal{H}_a$ are
\begin{itemize}[leftmargin=1.8cm]
\itemsep-2mm

   \item[$\bullet\hspace{2mm}\mathcal{H}_0 \hspace{2mm}$:] $f^{^{S}}_{v} = g^{^{\hspace{-0.3mm}S}}_{v} = 0~$;
   \item[$\bullet\hspace{2mm}\mathcal{H}_a\hspace{2mm}$:] $f^{^{S}}_{v} = g^{^{\hspace{-0.3mm}S}}_{v}\neq 0$. 
\end{itemize}

The details of the procedure are demonstrated as follows: 
\begin{enumerate}[leftmargin=*,labelindent=2mm,label= \textit{Step \arabic*}:]
\itemsep-2mm
    \item Estimate the values of the reflection symmetry test functions $\widehat{f}^{^{\hspace{0.5mm}S}}_{v_1}, \ldots, \widehat{f}^{^{\hspace{0.5mm}S}}_{v_m}, \widehat{g}^{^{\hspace{0.4mm}S}}_{w_1}, \ldots, \widehat{g}^{^{\hspace{0.4mm}S}}_{w_m}$ over the evenly spaced points $\{0\leq t_1, \ldots, t_p \leq 1\} $ as in Table~\ref{tb:tf} by $\bm{U}_i$'s, or $\widehat{\bm{U}}_i$'s, where for any $j \in \{1, \ldots, m\}$, $v_j$ and $w_j$ are randomly generated from the unit interval $[0,1]$. 
    
    \item Simulate from a reflection symmetric copula to obtain a set of $n$ observations, $\bm{V}^{^{\mathcal{H}_0}}_{i},\hspace{1mm} i \in \{ 1, \ldots, n\}$. Further details on how to simulate the observations $\bm{V}^{^{\mathcal{H}_0}}_{i}$ are described in Section~\ref{sim:h0}.
    
    \item Estimate the values of the reflection symmetry test functions $\widehat{f}^{^{\hspace{1mm}\mathcal{H}_0}}_{v_1}, \ldots, \widehat{f}^{^{\hspace{1mm} \mathcal{H}_0}}_{v_{m_0}}, \widehat{g}^{^{\hspace{1mm}\mathcal{H}_0}}_{w_1}, \ldots, \widehat{g}^{^{\hspace{1mm} \mathcal{H}_0}}_{w_{m_0}}$ over the evenly spaced points $\{0\leq t_1, \ldots, t_p \leq 1\}$ as in Table~\ref{tb:tf} by $\bm{V}^{\mathcal{H}_0}_{i}$'s, where for any $j^{0} \in \{1, \ldots, m_0\}$, $v_{j^{0}}$ and $w_{j^{0}}$ are randomly generated from the unit interval $[0,1]$.
    
    \item Combine the values of $\widehat{f}^{^{\hspace{0.5mm}S}}_{v_1}, \ldots, \widehat{f}^{^{\hspace{0.5mm}S}}_{v_m}, \widehat{g}^{^{\hspace{0.4mm}S}}_{w_1}, \ldots, \widehat{g}^{^{\hspace{0.4mm}S}}_{w_m}$ with those of $\widehat{f}^{^{\hspace{1mm}\mathcal{H}_0}}_{v_1}, \ldots, \widehat{f}^{^{\hspace{1mm} \mathcal{H}_0}}_{v_{m_0}}, \widehat{g}^{^{\hspace{1mm}\mathcal{H}_0}}_{w_1}, \ldots, \widehat{g}^{^{\hspace{1mm} \mathcal{H}_0}}_{w_{m_0}}$ to form a discretized functional data set of size $2(m + m_0)$. Compute their modified band depths \citep [see e.g.,][]{2009.LR.JASA,2012.SGN.Stat}. 
    
    \item Rank the $2(m + m_0)$ test functions according to their depth values. In case of any ties, we assign distinct ordinal numbers at random to the test functions that compare equal in terms of the depth values. Suppose that $\widehat{f}^{^{\hspace{0.5mm}S}}_{v_1}, \ldots, \widehat{f}^{^{\hspace{0.5mm}S}}_{v_m}, \widehat{g}^{^{\hspace{0.4mm}S}}_{w_1}, \ldots, \widehat{g}^{^{\hspace{0.4mm}S}}_{w_m}$ are associated with the ranks $r_1, \ldots, r_{2m_0}$. Define the test statistic $W = \sum_{i = 1}^{2m_0} r_i$. 
\end{enumerate}
The null hypothesis $\mathcal{H}_0$ is rejected when $W$ is significantly small because it means that the test functions $\widehat{f}^{^{\hspace{0.5mm}S}}_{v_1}, \ldots, \widehat{f}^{^{\hspace{0.5mm}S}}_{v_m}, \widehat{g}^{^{\hspace{0.4mm}S}}_{w_1}, \ldots, \widehat{g}^{^{\hspace{0.4mm}S}}_{w_m}$ are more deviated from zero. The definition of the test statistic $W$ here takes the essence of the Wilcoxon, or equivalently Mann-Whitney, test statistic. Hence, one can also deem the proposed hypothesis testing as a modification of the \textit{two-sample} Wilcoxon rank-sum test for one functional data set. The null distribution of $W$ is estimated by $N_b$ bootstrap samples of size $n$. More specifically, we generate $N_b$ samples from the reflection symmetry copula in \textit{Step 2}. For the $b$-th sample, $b \in \{1, \ldots, N_b\}$, we regard it as a set of $\bm{U}_i$'s and follow the above procedure to calculate the corresponding test statistic $W_b$. Eventually, the null distribution of $W$ is approximated by all of these test statistics: $W_1, \ldots, W_{N_b}$. For the details on how to carry out such a \textit{one-sample} bootstrapping method in hypothesis testing, we refer to Section 16.4 in~\cite{1993.ET.CHNY}.

Some features of the procedure are worth further discussion. First, the number of observations simulated from a reflection symmetry copula needs to be the same as the original sample size $n$. Otherwise, even if the test functions $\widehat{f}^{^{\hspace{0.5mm}S}}_{v_1}, \ldots, \widehat{f}^{^{\hspace{0.5mm}S}}_{v_m}, \widehat{g}^{^{\hspace{0.4mm}S}}_{w_1}, \ldots, \widehat{g}^{^{\hspace{0.4mm}S}}_{w_m}$ came from a reflection symmetry copula, they would still be more centered/deviated with respect to zero, thereby having greater depth values, compared to $\widehat{f}^{^{\hspace{1mm}\mathcal{H}_0}}_{v_1}, \ldots, \widehat{f}^{^{\hspace{1mm} \mathcal{H}_0}}_{v_{m_0}}, \widehat{g}^{^{\hspace{1mm}\mathcal{H}_0}}_{w_1}, \ldots, \widehat{g}^{^{\hspace{1mm} \mathcal{H}_0}}_{w_{m_0}}$. This is because larger or smaller sample sizes make empirical copulas better or worse approximations of the underlying true copulas. 

Second, the values of $p$, $m$, and $m_0$ in \textit{Steps 1} and \textit{3} are hyperparameters that one can tune at will. Nonetheless, it is observed that setting $p = 100$ is sufficient, while $m$ and $m_0$ have better to be large to guarantee the ideal empirical size and power of the test. Particularly, we notice that having $m$ and $m_0$ greater than $n$ yields better performances. Once both $m$ and $m_0$ satisfy the property, their values tend to have little effect on its performance.  

Third, the most notable difference between our procedure and those used in \cite{2009.LR.JASA},  \cite{2019.HS.SS} and \cite{2023.HSG.JCGS} lies in the absence of a reference data set. The ideas in their approaches originate with the hypothesis testing procedure designed by~\cite{1993.LS.JASA} that performs the detection of (dis)similarity in two multivariate distributions using \textit{quality index}. Particularly, they introduced the treatment with a reference data set to identify the change in population locations. In the context of functional data, it describes the scenario when the two samples of functional curves appear to form two separated bands, each of which comprises only curves from one group as shown in Figure~\ref{fig:reason4rd}. However, it should not be a concern in our testing procedure thanks to the introduction of $\widehat{g}^{^{\hspace{0.4mm}S}}_{w_1}, \ldots, \widehat{g}^{^{\hspace{0.4mm}S}}_{w_m}$ which are rough reflections of the test functions $\widehat{f}^{^{\hspace{0.4mm}S}}_{v_1}, \ldots, \widehat{f}^{^{\hspace{0.4mm}S}}_{v_m}$ with respect to zero. As one may see later, we actually generate the observations $\bm{V}^{\mathcal{H}_0}_{i}, i \in \{1, \ldots, n\}$ from a mixture ($M$) of empirical copulas in \textit{Step 2} under the assumption that $M$ is close to a truly reflection symmetry copula. Then, if we had to simulate a reference data set, the empirical copula of the reference data set might be more symmetric than that of the observations $\bm{V}^{\mathcal{H}_0}_{i}, i \in \{1, \ldots, n\}$ when the sample size $n$ is not sufficiently large. This could lead to a not-so-small test statistic in the approach of~\cite{2019.HS.SS} or~\cite{2023.HSG.JCGS} as the test functions constructed from the reference data set might tend to be slightly more centered around zero. Certainly, the absence of a reference data set also saves memory storage. 

\begin{figure}[htb!]
   \caption{ In the context of functional data, the approach of \cite{1993.LS.JASA} can be interpreted as the scenario where two samples of functional curves exhibit a distinct pattern. Specifically, the curves from each group form separate bands, with little to no overlap between the two groups.}
    \centering
    \includegraphics[height=7.5cm, width=0.6\textwidth]{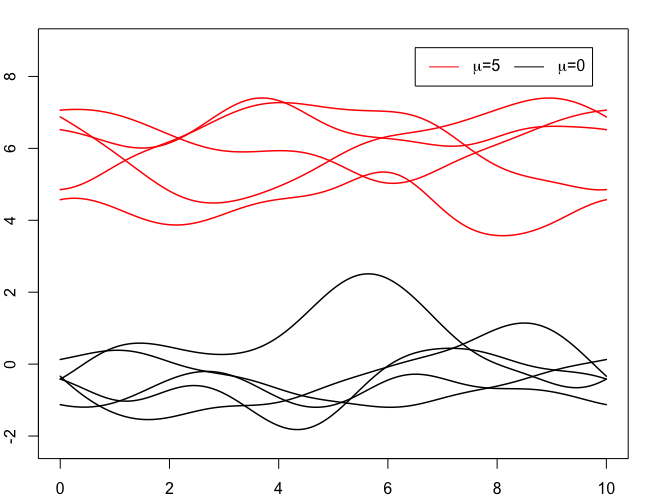}
    \label{fig:reason4rd}
\end{figure}

\subsubsection{Simulation under $\mathcal{H}_0$}\label{sim:h0}
Lastly, it is important to address the choice of the reflection symmetric copula in Step 2 and explain how to simulate from it. In order to maintain the nonparametric nature of the procedure, it is crucial to avoid selecting any parametric copula. In lieu, we propose to construct the estimated reflection symmetric copula $\widehat{C}^{S}$ given, for all $(u, v)$, by either
\begin{equation*}
    \widehat{C}^{S} (u, v) = \frac{1}{2} \Big ( \widehat{C}_n (u, v) + \widehat{C}_n (v, u)\Big )
\end{equation*}
or 
\begin{equation*}
    \widehat{C}^{S} (u, v) = \frac{1}{2} \Big ( \widehat{D}_n (u, v) + \widehat{D}_n (v, u)\Big ),
\end{equation*}
depending on the types of data set we have.

\begin{property}\label{prop:symc} 
	Given any bivariate copula $C$, the copula $\tilde{C}$ defined, for all $(u, v) \in [0, 1]^2$, by 
	\begin{equation*}
		\tilde{C}(u, v) = \frac{1}{2} \Big ( C (u, v) +  C(v, u) \Big )
	\end{equation*}
	is a reflection symmetric copula. 
\end{property}
\noindent We prove this result in the Appendix~\ref{A:prop4}.

As a mixture of two empirical copulas, one should have no difficulty simulating from $\widehat{C}^{S}$. Note that $\widehat{C}^{S}$ is not really a copula itself, but it serves as an appropriate estimation of the reflection symmetric copula $\tilde{C}$ \citep{Ruschendorf1976}. Apart from the intuition suggested by Proposition~\ref{prop:symc},  Figure~\ref{fig:simulated_copula_1} provides a visual assurance for the simulation method. 

\begin{figure}[b!]
    \centering
    \includegraphics[width=\textwidth]{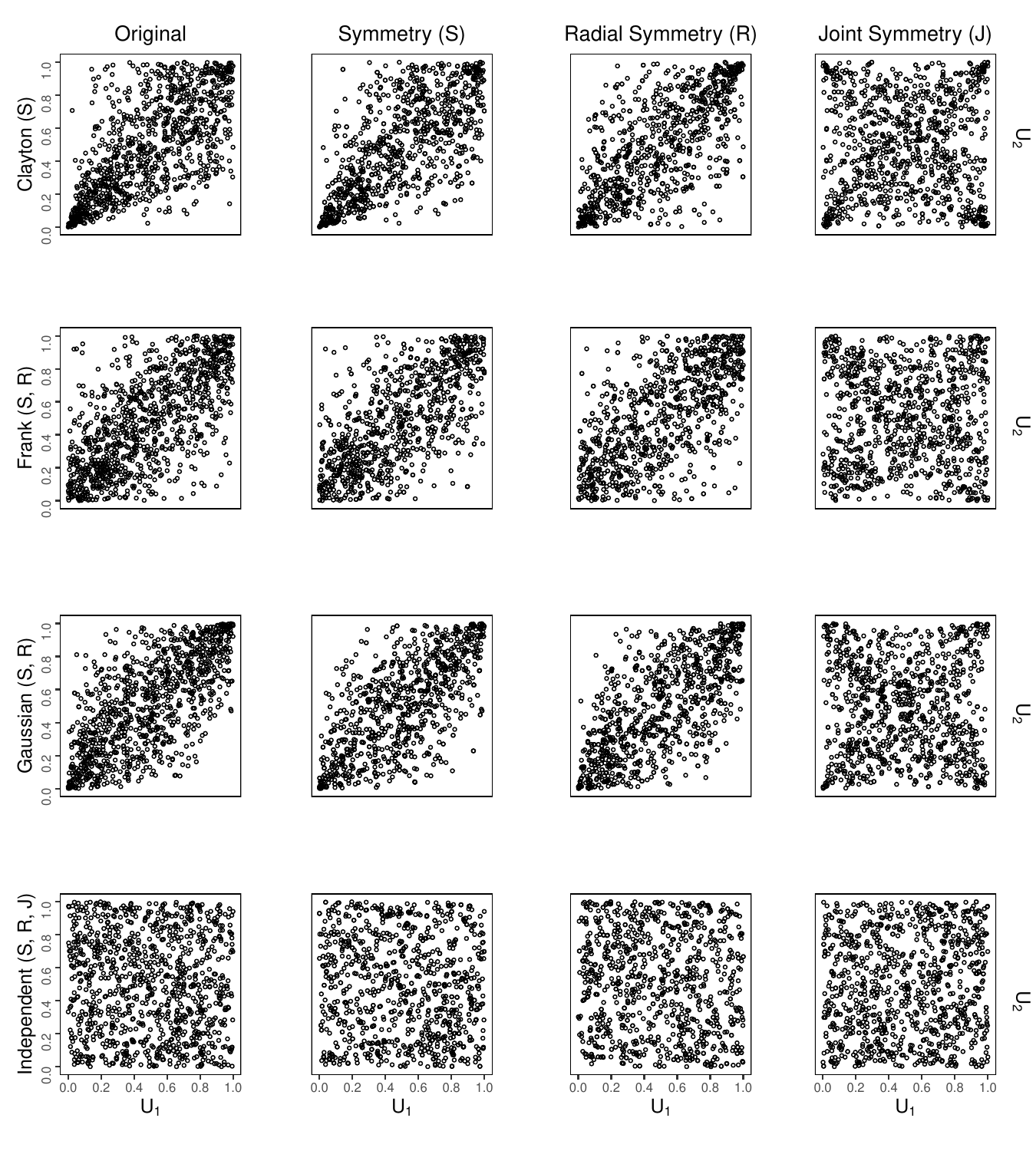}
\end{figure}
\begin{figure}[ht!]
    \centering
    \includegraphics[width=\textwidth]{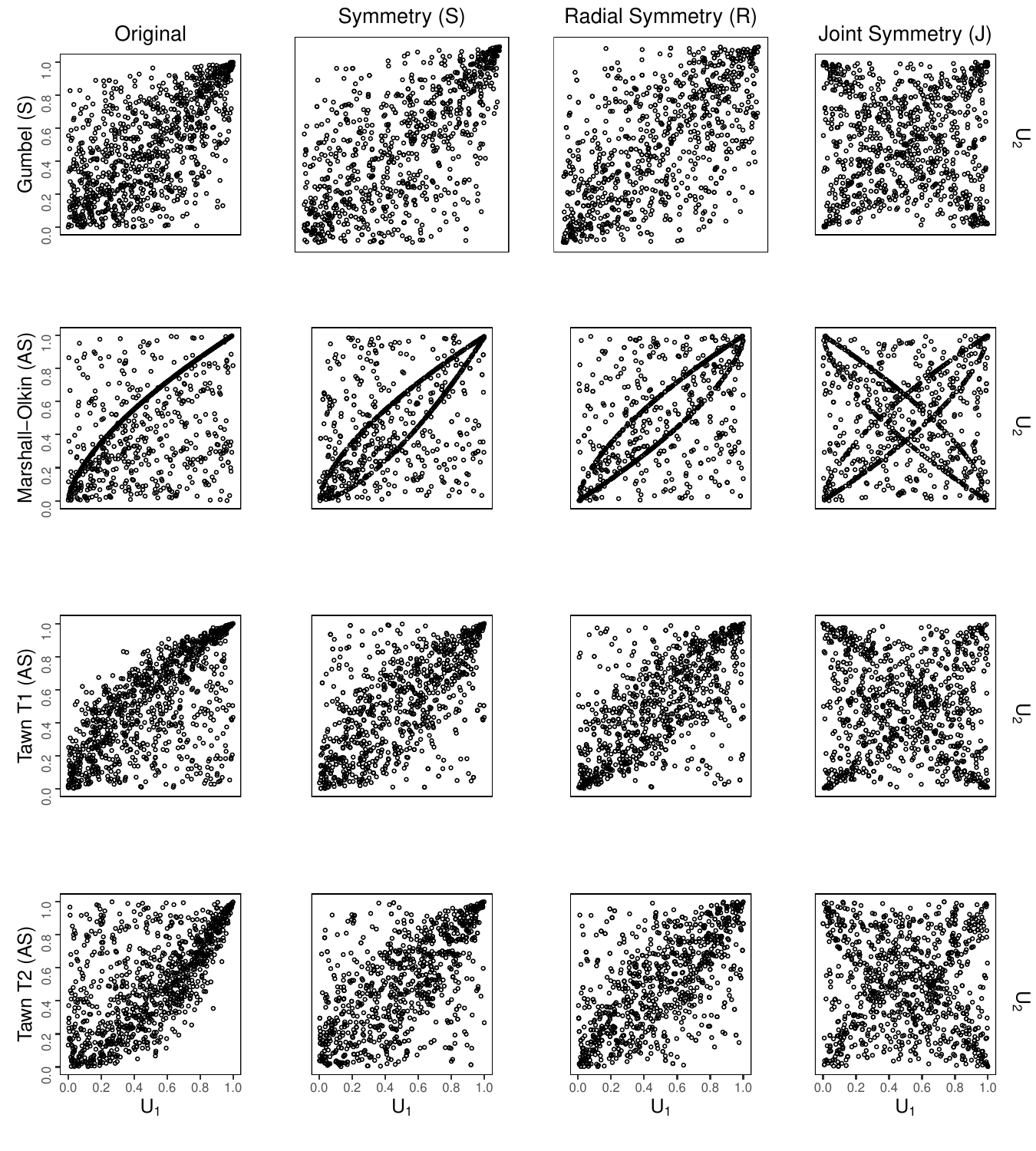}
    \caption{Rank plots for the different copula models considered in Section~\ref{Vis:FB}}
    \label{fig:simulated_copula_1}
\end{figure}

By closely following the aforementioned procedure and making appropriate adjustments, one can conduct tests for radial and joint symmetries. Specifically, in Step 2, a copula that exhibits radial or joint symmetry needs to be constructed. Table~\ref{tb:cpls} provides a summary of the proposed estimated copulas for this purpose, and the following Propositions provide justification for these choices. To ensure clarity in our notation, we denote the survival copula of any copula $C$ as $C^\ast$. Proofs of the Propositions~\ref{prop:rsymc} and \ref{prop:jsymc} can be found in the Appendices~\ref{A:prop5} and \ref{A:prop6}, respectively.

\begin{property} \label{prop:rsymc}
	Given any bivariate copula $C$, the copula $\tilde{C}$ defined, for all $(u, v) \in [0, 1]^2$, by 
	\begin{equation*}
		\tilde{C}(u, v) = \frac{1}{2} \Big ( C(u, v) + C^\ast(u, v) \Big)
	\end{equation*}
	is a radially symmetric copula. 
\end{property}

\begin{property} \label{prop:jsymc}
	Given any copula $C$, the copula $\tilde{C}$ defined, for all $(u, v) \in [0, 1]^2$, by 
	\begin{equation*}
		\tilde{C}(u, v) = \frac{1}{4} \Big (C(u, v) + u - C(u, 1- v) + v - C(1-u, v) + C^\ast(u, v) \Big )
	\end{equation*}
	is a jointly symmetric copula. 
\end{property}

\begin{table}[htb!]
\centering
    \begin{tabular}{c|c|c}
    \hline
    \\[-0.8em]
         \textbf{Type} & \textbf{Access} & \textbf{Simulated Distribution} \\
         \\[-0.8em]
         \hline
         \\[-1em]
         
         \multirow{4}{*}{S} & $(U_1, U_2)$ & $\widehat{C}^{S} (u, v) = \frac{1}{2} \Big ( \widehat{C}_n (u, v) + \widehat{C}_n (v, u)\Big )$\\[1.5ex]
        \cmidrule{2-3}
        
         & $(X_1, X_2)$ & $\widehat{C}^{S} (u, v) = \frac{1}{2} \Big ( \widehat{D}_n (u, v) + \widehat{D}_n (v, u)\Big )$\\[1.5ex]
         \hline
         \\[-1em]
         
         \multirow{4}{*}{R} & $(U_1, U_2)$ &  $\widehat{C}^{R} (u, v) = \frac{1}{2} \Big ( \widehat{C}_n (u, v) + \widehat{C}^\ast_n (u, v)\Big )$ \\[1.5ex]
        \cmidrule{2-3}
        
        & $(X_1, X_2)$ & $\widehat{C}^{R} (u, v) = \frac{1}{2} \Big ( \widehat{D}_n (u, v) + \widehat{D}^\ast_n (u, v)\Big )$\\[1.5ex]
         \hline
         \\[-1em]
         
        \multirow{4}{*}{J} & $(U_1, U_2)$ & $ \widehat{C}^{{J}}(u, v) = \hspace{1mm}\tfrac{1}{4} \Big ( \widehat{C}_n(u, v) + u - \widehat{C}_n(u, 1 - v) 
         + v - \widehat{C}_n(1-u, v) + \widehat{C}^\ast_n(u, v ) \Big)$ \\[1.5ex]
         \cmidrule{2-3}

         & $(X_1, X_2)$ & $ \widehat{C}^{{J}}(u, v) = \hspace{1mm}\tfrac{1}{4} \Big ( \widehat{D}_n(u, v) + u - \widehat{D}_n(u, 1 - v) 
         + v - \widehat{D}_n(1-u, v) + \widehat{D}^\ast_n(u, v ) \Big) $ \\[1.5ex]
         \hline
    \end{tabular}
    \caption{Definitions of estimated copulas constructed to test different types of symmetry. $\widehat{C}^\ast_n$ and $\widehat{D}^\ast_n$ denote the estimator for the survival copulas associated with the underlying copula.}
    \label{tb:cpls}
\end{table}



\section{Simulation Study}\label{Simulation}

We performed simulations to evaluate the effectiveness of our testing procedures for the three copula structures described in Section~\ref{cop_sym}. In each simulation, we estimated 1200 test functions, and the tests were conducted at a nominal level of $5\%$. We utilized $N_b=1000$ bootstrap samples for our analysis. To assess the performance of our approach for testing reflection symmetry, we compared it to the tests proposed by \cite{2012.GNQ.AISM} and \cite{2013.LG.JASA}.  The results, including the sizes and powers of the test for reflection symmetry, are summarized in Table~\ref{tab:symmetry}. To evaluate the power of our tests, we introduce asymmetry to the copula models using Khoudraji's device \citep{Khoudraji1995}. Specifically, we used an asymmetric version of a copula $C(u,v)$, defined for $(u,v) \in [0,1]^2$, given by:
\begin{equation*}
K_{\delta}(u,v) = u^{\delta}C(u^{1-\delta},v),
\end{equation*}
\noindent where $\delta \in (0,1)$. Previous studies \citep{2012.GNQ.AISM} have shown that Khoudraji's device introduces minimal asymmetry when Kendall's $\tau \leq 0.5$. The maximum level of asymmetry is typically observed around $\delta=0.5$. The R codes for implementing the visualization and hypotheses testing procedures are available at \url{https://github.com/cfjimenezv07/Visualization-and-Assessment-of-Copula-Symmetry}.

\begin{table}[!htb]
    \centering
    \addtolength{\leftskip} {-2.2cm}
    \addtolength{\rightskip}{-2cm}
     \resizebox{\textwidth}{!}{
    \begin{tabular}{ccccccc|cccc|cccc}
    \toprule
    \toprule
         \multicolumn{3}{c}{}&
         \multicolumn{4}{c}{Clayton} &
         \multicolumn{4}{c}{Gaussian} &
         \multicolumn{4}{c}{Gumbel}\\
         
         \multicolumn{3}{c}{} &
         \multicolumn{4}{c}{$n$} &
         \multicolumn{4}{c}{$n$} &
         \multicolumn{4}{c}{$n$}\\ \cmidrule{4-7} \cmidrule{8-11} \cmidrule{12-15}
         
         $\delta$ & $\tau$ &  & 100 & \hl 250 & \hl 500 & \hs 1000 & 100 & \hl 250 & \hl 500 & \hs 1000 & 100 & \hl 250 & \hl 500 & \hs 1000 \\
         
         \midrule
         \multirow{3}{*}{0}
         & 1/4 & \multirow{4}{*}{\rotatebox{90}{\rotatebox{-90}{E} \rotatebox{-90}{Z} \rotatebox{-90}{\hspace{0.4mm}I} \rotatebox{-90}{S} \hspace{2mm}} }
         & 0.044 & 0.048 & 0.059 & 0.062 & 0.053 & 0.052 & 0.046 & 0.066 & 0.050 & 0.063 & 0.059 & 0.063 \\
         
         & 1/2 &
         & 0.021 & 0.043 & 0.051 & 0.053 & 0.031 & 0.039 & 0.048 & 0.064 & 0.038 &
         0.044 & 0.055 & 0.053\\
         
         & 3/4 &
         & 0.007 & 0.002 & 0.007 & 0.007 & 0.001 & 0.002 & 0.003 & 0.004 & 0.007 &
         0.008 & 0.009 & 0.019\\
         \cmidrule{1-2} \cmidrule{4-15}
         
         0 & $0^\ast$ &  
         & 0.045 & 0.060 & 0.059 & 0.060 & \hl \xdash & \hl \xdash & \hl \xdash & \hl \xdash & \hl \xdash & \hl \xdash & \hl \xdash & \hl \xdash \\
         
         \midrule
         
         \multirow{3}{*}{1/4} & 0.5 &  \multirow{9}{*}{\rotatebox{90}{
         \rotatebox{-90}{\hspace{0.4mm}R}
         \rotatebox{-90}{\hspace{0.4mm}E} 
         \rotatebox{-90}{W} \rotatebox{-90}{\hspace{0.4mm}O} \rotatebox{-90}{\hspace{0.4mm}P} \hspace{2mm}} }
         & 0.121 & 0.297 & 0.547 & 0.796 & 0.089 & 0.267 & 0.552 & 0.850 & 0.119 &
         0.312 & 0.609 & 0.903\\
         
         & 0.7 & &
         0.374 & 0.814 & 0.981 & 0.999 & 0.336 & 0.860 & 0.995 & 1.000 & 0.365 & 
         0.886 & 0.994 & 1.000\\
         
         & 0.9 & 
         & 0.675 & 0.977 & 0.996 & 1.000 & 0.740 & 0.987 & 0.998 & 1.000 & 0.728 & 
         0.990 & 1.000 & 1.000\\
         \cmidrule{1-2} \cmidrule{4-15}
         
         \multirow{3}{*}{1/2} & 0.5 &
         & 0.149 & 0.358 & 0.630 & 0.872 & 0.187 & 0.511 & 0.826 & 0.990 & 0.264 &
         0.690 & 0.929 & 0.998\\
         
         & 0.7 & 
         & 0.483 & 0.926 & 0.998 & 1.000 & 0.664 & 0.986 & 1.000 & 1.000 & 0.736 &
         0.992 & 1.000 & 1.000\\
         
         & 0.9 & 
         & 0.908 & 1.000 & 1.000 & 1.000 & 0.951 & 1.000 & 1.000 & 1.000 & 0.940 & 1.000 & 1.000 & 1.000\\
         \cmidrule{1-2} \cmidrule{4-15}
         
         \multirow{3}{*}{3/4} & 0.5 &
         & 0.087 & 0.210 & 0.334 & 0.557 & 0.178 & 0.450 & 0.743 & 0.959 & 0.301 & 0.656 & 0.930 & 0.995\\
         
         & 0.7 & &
         0.254 & 0.629 & 0.897 & 0.996 & 0.500 & 0.929 & 0.999 & 1.000 & 0.605 &
         0.957 & 0.998 & 1.000\\
         
         & 0.9 & & 
         0.662 & 0.984 & 1.000 & 1.000 & 0.761 & 0.999 & 1.000 & 1.000 & 0.768 &
         0.985 & 1.000 & 1.000\\
         \midrule
         \midrule
    \end{tabular}}
    \caption{Empirical sizes and powers of the test of symmetry in the setting as \cite{2012.GNQ.AISM} and \cite{2013.LG.JASA}. The sizes for the independent copula are $\delta=\tau=0$.}
    \label{tab:symmetry}
\end{table}

Our results indicate that under small and moderate $\tau$, the sizes converge to the nominal value as the sample size increases, but with larger $\tau$, the sizes are somewhat below the nominal level. The powers of the tests increase as the sample size increases. When comparing our results to Table 3 in \cite{2012.GNQ.AISM}, their approach generally achieves better powers for a small sample size of $n=250$. In contrast, our approach demonstrates significantly improved powers as the sample size increases. This is expected as our rank-based test relies on the asymptotic distribution of the test functions and the empirical copula process.

When comparing our simulation results to Table 1 in \cite{2013.LG.JASA}, we find that our results align with most of the reported powers and sizes. However, our approach achieves considerably higher powers even at smaller sample sizes, particularly for intermediate values of $\tau$ where the maximum asymmetry is expected.

For the tests of radial and joint symmetry, we selected five commonly used copulas. The sizes and powers of both testing procedures are presented in Table~\ref{tab:radsymmetry}. To evaluate the performance of our approach in testing radial symmetry, we compared it to the methods proposed by \cite{2013.LG.JASA} and \cite{2014.GN.SP}. Additionally, we compared our approach for testing joint symmetry to the methods reported by \cite{2013.LG.JASA}.

\begin{table}[!htb]
    \centering
    \resizebox{\textwidth}{!}{
    \begin{tabular}{ccccccc|ccccc}
    \toprule
    \toprule
         \multicolumn{2}{c}{}&
         \multicolumn{5}{c}{Radial} &
         \multicolumn{5}{c}{Joint}\\
         \cmidrule{3-7}
         \cmidrule{8-12}
         
          & $\tau$ & $n=$ & $100$ &
         $250$ & $500$ & $1000$ & $n =$ & $100$ & $250$ &  $500$ & $1000$ \\
         \midrule
         
         $\Pi$ & $0$ & 
         \multirow{7}{*}{\rotatebox{90}{\rotatebox{-90}{E} \rotatebox{-90}{Z} \rotatebox{-90}{\hspace{0.4mm}I} \rotatebox{-90}{S} \hspace{2mm}} }
         & 0.060 & 0.073 & 0.065 & 0.060& \scriptsize SIZE & 0.030 & 0.029 & 0.024 & 0.048 \\
         \cmidrule{1-2}
         \cmidrule{4-7}
         \cmidrule{8-12}
         
         \multirow{3}{*}{Frank}
         & 1/4 &
         & 0.054 & 0.060 & 0.059  & 0.054 & \multirow{12}{*}{\rotatebox{90}{
         \rotatebox{-90}{\hspace{0.4mm}R}
         \rotatebox{-90}{\hspace{0.4mm}E} 
         \rotatebox{-90}{W} \rotatebox{-90}{\hspace{0.4mm}O} \rotatebox{-90}{\hspace{0.4mm}P} \hspace{2mm}}} & 0.685 & 0.963 & 0.997 & 1.000\\
         
         & 1/2 &
         & 0.050 & 0.064 & 0.058 & 0.059 & & 0.968 & 1.000 & 1.000 & 1.000 \\
         
         & 3/4 &
         & 0.016 & 0.041 & 0.037 & 0.046 & & 0.975 & 0.999 & 1.000 & 1.000 \\
         \cmidrule{1-2}
         \cmidrule{4-7}
         \cmidrule{9-12}
         
        \multirow{3}{*}{Gaussian}
         & 1/4 &
         & 0.062 & 0.060 & 0.047 & 0.059 &  & 0.725 & 0.984 & 1.000 & 1.000 \\
         
         & 1/2 &
         & 0.041 & 0.041 & 0.054 & 0.054 & & 0.976 & 0.999 & 1.000 & 1.000 \\
         
         & 3/4 &
         & 0.009 & 0.022 & 0.014 & 0.031 & & 0.973 & 1.000 & 1.000 & 1.000 \\
         \cmidrule{1-2}
         \cmidrule{3-7}
         \cmidrule{9-12}
         
        \multirow{3}{*}{Clayton}
         & 1/4 & \multirow{6}{*}{\rotatebox{90}{
         \rotatebox{-90}{\hspace{0.4mm}R}
         \rotatebox{-90}{\hspace{0.4mm}E} 
         \rotatebox{-90}{W} \rotatebox{-90}{\hspace{0.4mm}O} \rotatebox{-90}{\hspace{0.4mm}P} \hspace{2mm}}} 
         & 0.228 & 0.491 & 0.830 & 0.985 &  & 0.706 & 0.957 & 0.999 & 1.000 \\
         
         & 1/2 &
         & 0.516 & 0.833 & 0.985 & 1.000 & & 0.966 & 0.999 & 1.000 & 1.000 \\
         
         & 3/4 &
         & 0.291 & 0.512 & 0.912 & 0.996 & & 0.976 & 1.000 & 1.000 & 1.000 \\
         \cmidrule{1-2}
         \cmidrule{4-7}
         \cmidrule{9-12}
         
         \multirow{3}{*}{Gumbel}
         & 1/4 &  
         & 0.110 & 0.333 & 0.510 & 0.695 &  & 0.696 & 0.974 & 0.998 & 1.000 \\
         
         & 1/2 &
         & 0.166 & 0.486 & 0.723 & 0.909 & & 0.974 & 1.000 & 1.000 & 1.000 \\
         
         & 3/4 &
         & 0.042 & 0.322 & 0.580 & 0.815 & & 0.978 & 0.999 & 1.000 & 1.000 \\
         
         \midrule
         \midrule
    \end{tabular}}
\caption{Sizes and powers of the test of radial and joint symmetry in the setting as~\cite{2013.LG.JASA} and~\cite{2014.GN.SP}. One bootstrap sample is used to estimate the distribution of the test statistic $W$ under the null.}
    \label{tab:radsymmetry}
\end{table}

The sizes of both tests closely align with the nominal level of $5\%$ for values of $\tau$ equal to $1/4$ and $1/2$. However, for larger values of $\tau$, the sizes are achieved at larger sample sizes. Specifically, in the case of the radial symmetry test, our results are consistent with those presented in Table 2 of \cite{2013.LG.JASA}, and once again our testing approach achieves the nominal levels at smaller sample sizes.

In comparison to the results presented in Table 1 of \cite{2014.GN.SP} for the Frank and Gaussian copula models they considered, our approach achieves sizes that are closer to the nominal level. Although their performance improves for larger sample sizes, our approach outperforms them even in scenarios with larger sample sizes.

Regarding the joint symmetry copula structure, our achieved powers are significantly higher compared to the powers obtained for the radial symmetry property, similar results are presented in Table 2 of \cite{2013.LG.JASA}.


\section{Data Applications}\label{data_app}
To illustrate the procedures herein, we apply our visualization and hypothesis testing methodology to two real-world data sets. The Nutritional Habits Survey Data is in Section~\ref{app:NHSD}, and the wind speed dataset in Saudi Arabia is in Section~\ref{app:WD}.

\subsection{Nutritional Habits Survey Data}\label{app:NHSD}

The dataset used here comes from a survey conducted by the U.S. Department of Agriculture in 1985. The survey aimed to investigate the dietary habits of 737 women aged between 25 and 50 years. Specifically, the survey collected daily intake measurements of five variables: calcium (mg), iron (mg), protein (g), vitamin A (mg), and vitamin C (mg).

In previous analyses, \cite{2012.GNQ.AISM} employed a Cram\'er-von Mises statistic to evaluate the bivariate reflection symmetry of the pairwise copulas. Similarly, \cite{2013.LG.JASA} utilized a nonparametric method based on the asymptotic distribution of the empirical copula process to assess reflection symmetry for this dataset. In our study, we conducted a test using our rank-based method and obtained corresponding p-values, which are presented in the top-right corner of each subplot in Figure~\ref{fig:nutrition}.

The pattern of p-values obtained in our study exhibits similarities to the findings reported by \cite{2012.GNQ.AISM} and \cite{2013.LG.JASA}, leading to similar conclusions for some pairs. However, there are notable differences. Unlike \cite{2012.GNQ.AISM}, our test does not reject the reflection symmetry copula at a $5\%$ significance level for the pairs (iron, vitamin A), (iron, vitamin C), and (protein, vitamin A). Conversely, we do reject the reflection symmetry copula for the pair (protein, vitamin C).

When comparing our results to those of \cite{2013.LG.JASA}, we find much closer agreement. The only difference in conclusions arises for the two pairs: (protein, vitamin C), where we reject the reflection symmetry copula, and (iron, vitamin A), where we do not reject it. Notably, the pair (iron, vitamin A) differs from the conclusions of both \cite{2012.GNQ.AISM} and \cite{2013.LG.JASA}. However, our proposed visualization method reveals that the majority of the test functions exhibit high density around the zero level for this pair.

\begin{figure}[!b]
    \centering
    \includegraphics[width=\textwidth]{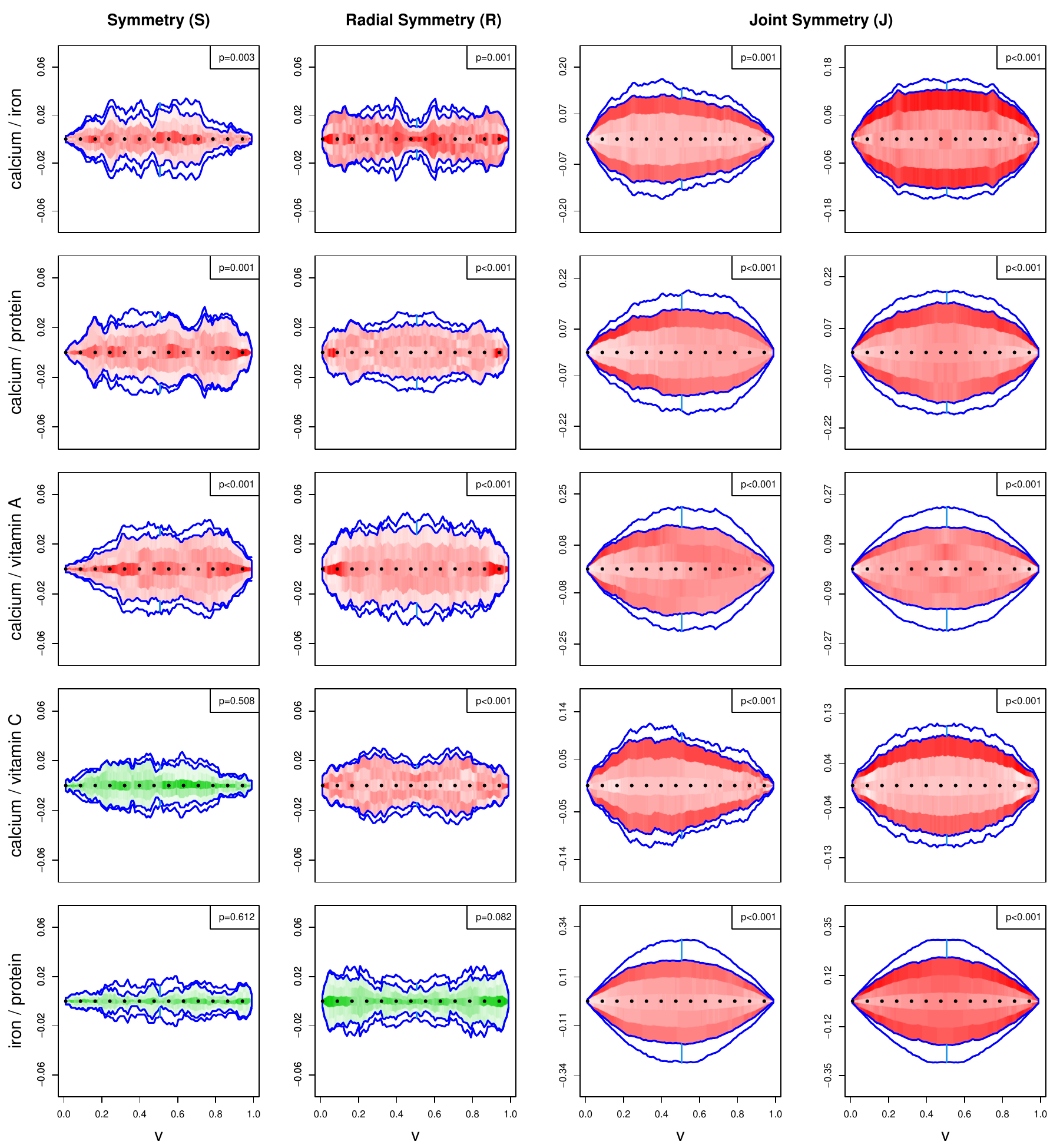}
\end{figure}
\begin{figure}[!ht]
    \centering
    \includegraphics[width=\textwidth]{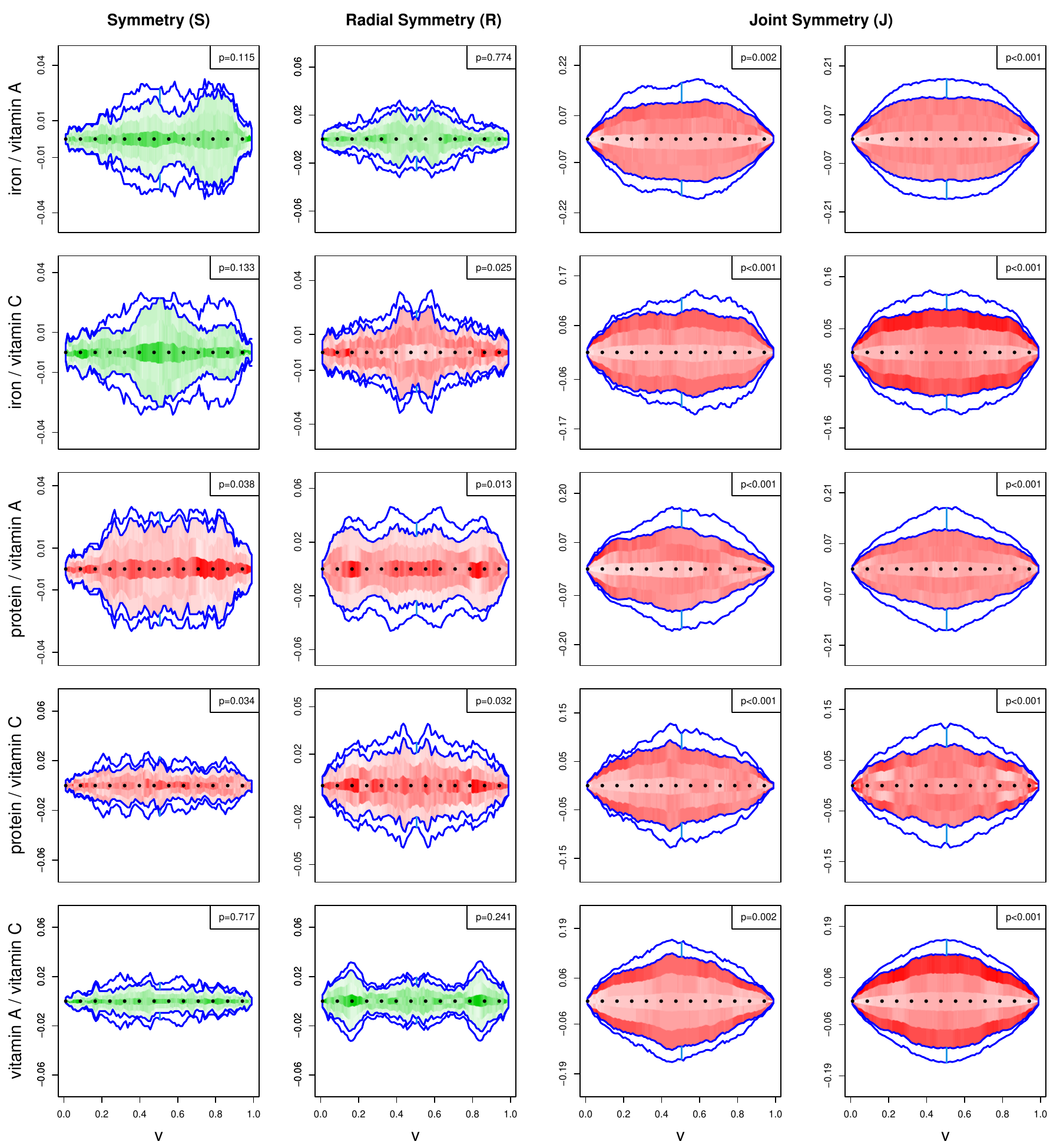}
    \caption{Visualizations and testing results for copula symmetries over the $5$ random variables in the nutritional data set.}
    \label{fig:nutrition}
\end{figure}

\subsection{Wind data}\label{app:WD}

In this study, we analyze a trivariate wind speed dataset obtained from \cite{Yip}. The dataset comprises bi-weekly mid-day wind speed measurements recorded over the period of 2009-2014 at three specific locations near Dumat Al Jandal, the site of the first wind farm currently under construction in Saudi Arabia. The dataset consists of $n = 156$ trivariate wind speed vectors, each containing measurements from three positions.

Understanding the dependence structure within these trivariate wind speed vectors is crucial for assessing how wind patterns impact the electricity generation of the nearby wind farm. By analyzing this dependence structure, we can gain insights into the interactions between wind speeds at different locations and optimize the operation and output of the wind farm.

One specific aspect of interest is evaluating whether a potential copula model, which quantifies the dependence among bivariate positions, exhibits a specific symmetry structure. This assessment allows us to examine the symmetry properties of a potential copula model and determine if it aligns with the desired symmetry structure. Understanding the symmetry properties can be essential for accurately modeling and predicting wind speed behavior.

Figure~\ref{fig:WD} presents the visualization and hypothesis testing results for the dataset under consideration, focusing on the three copula structures analyzed in the paper. The results show that joint symmetry is rejected for all pairs of locations. However, reflection symmetry and radial symmetry are not rejected by the bivariate dataset formed by positions 2 and 3. This is indicated by the functional boxplot, which demonstrates functions distributed widely around zero and the corresponding large p-values. In summary, our approach suggests using asymmetric copula models to properly capture the dependence in these wind speed data.

\begin{figure}[!htb]
    \centering
\includegraphics[width=\textwidth]{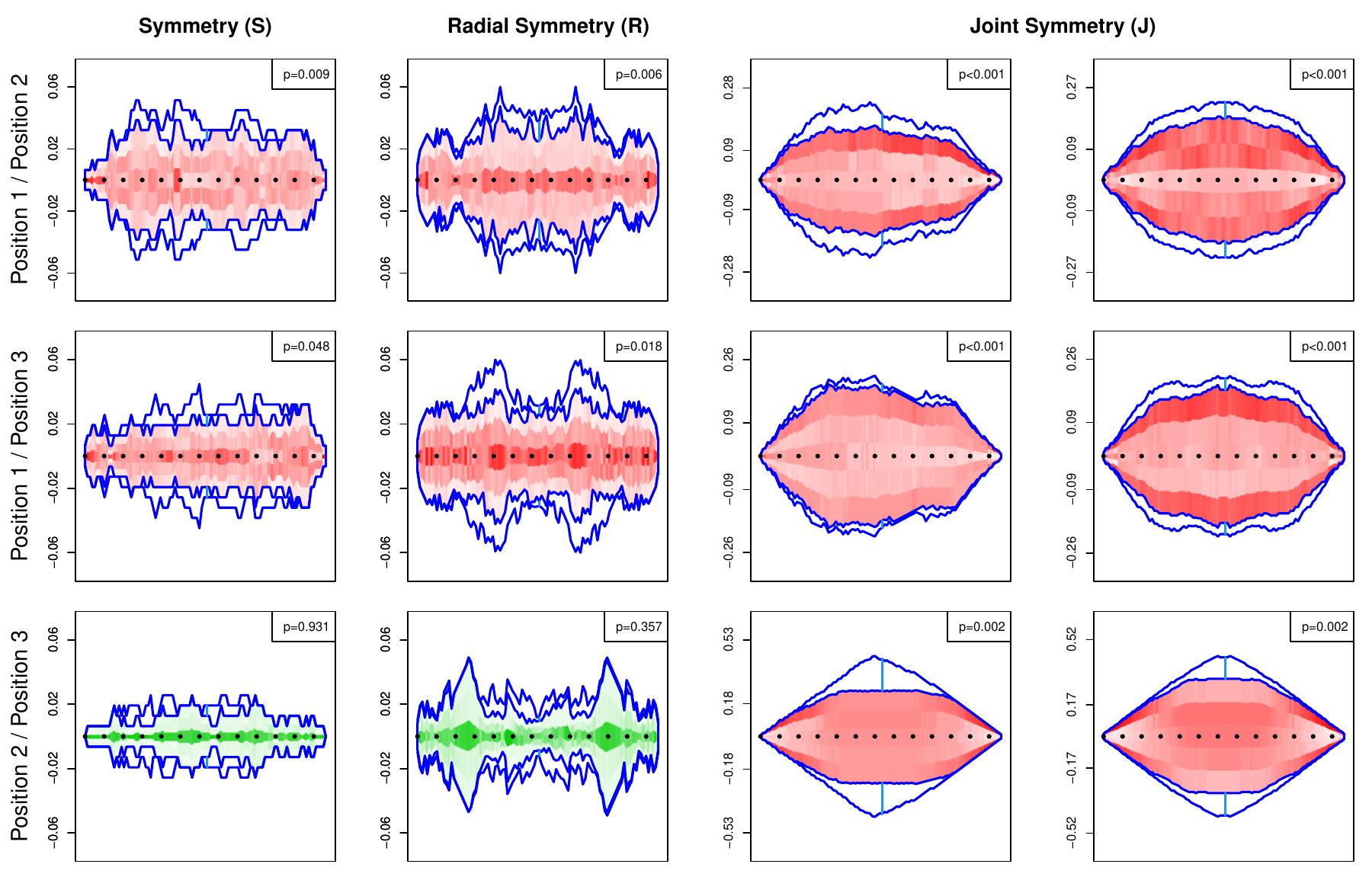}
      \caption{Visualizations and testing results for copula symmetries over the $3$ locations for the wind data in Saudi Arabia.}
    \label{fig:WD}
\end{figure}


\section{Discussion}\label{Discuss}
In this paper, we have presented a comprehensive framework for visualizing and testing common assumptions regarding copula structures, such as reflection symmetry, radial symmetry, and joint symmetry. Our approach utilizes functional data analysis techniques to construct test functions based on bivariate copulas at specific discrete points. These test functions effectively summarize the copula structures and provide valuable insights into their adherence to specific structures.

To visually represent the copula structures, we employ functional boxplots, which depict the functional median and variability of the test functions. These visualizations allow us to assess the departure from zero and gain insights into the degree to which the copula structures conform to the desired assumptions.

Additionally, we have introduced a nonparametric testing procedure to evaluate the significance of deviations from symmetry. Through extensive simulation studies involving various copula models, we have demonstrated the reliability and power of our method, particularly for moderate to large sample sizes. The numerical experiments and data analyses conducted in this study have consistently shown robust testing results across different datasets, and the visualization technique has proven useful in extracting preliminary information directly from the data. It is worth mentioning that our functional data approach can be extended to test copula properties beyond symmetry as well.

It is important to note that our testing method relies on the asymptotic distribution of the estimators of the test functions, which are derived from empirical copula processes. As a result, the small sample properties of our test may not be as optimal as certain existing testing methods. The required sample size for our testing approach can vary depending on the specific copula structure under examination, but in general, larger sample sizes tend to enhance the size and power of the test. Thus, increasing the sample size is recommended to improve the overall performance of the test in terms of accuracy and sensitivity.

Finally, our visualization and testing techniques were applied to two real-world datasets: a nutritional habits survey with five variables and wind speed data from three locations in Saudi Arabia. These applications provided valuable insights into the underlying structures and patterns within the datasets, demonstrating the effectiveness of our approach in gaining a better understanding of the data.


\baselineskip=15pt
\bibliographystyle{chicago}
\bibliography{Paper3}

\appendix

\section{Appendix}\label{appendix}
\subsection{\textit{Proof of Proposition~\ref{prop:a1}}}\label{A:prop1}
 The proof of the proposition is essentially the same as that of Proposition 2 in~\cite{2012.GNQ.AISM}. They focus on the asymptotic results when the underlying copula $C$ is indeed reflection symmetric. We firstly consider the case when $\widehat{f}^{^{\hspace{0.5mm}S}}_v (t) = \widehat{C}_n(t, v) - \widehat{C}_n (v, t)$. The map
\begin{equation*}  
    \mathcal{T}^S_v: \ell([0, 1]^2) \rightarrow \ell([0, 1]) : f \mapsto \big (t \mapsto f(t, v) - f(v, t)\big)\\
\end{equation*}
is a continuous function. Thus, the Continuous Mapping Theorem (\cite{1996.VW.SVNY}, Theorem 1.3.6) and result~\eqref{eq:ec1} imply that $\sqrt{n} \left \{ \widehat{f}^{^{\hspace{0.5mm}S}}_v (t) - f^{^{S}}_v (t) \right \} \cid \mathbb{E}^{^S}_v (t) := \mathbb{C} (t, v) - \mathbb{C}(v, t)$ for all $t \in [0, 1]$, where $\mathbb{E}^{^S}_v$ is a centered Gaussian random field with covariance function given at each $s, t \in [0, 1]$ by 
\begin{align*}
    \Cov{\mathbb{E}^{^S}_v (t)}{\mathbb{E}^{^S}_v (s)} = & \hspace{2mm} \Cov{\mathbb{C}(t, v) - \mathbb{C}(v, t)}{\mathbb{C}(s, v) - \mathbb{C}(v, s)}\\
    = & \hspace{2mm} C(t \land s, v) - C(t, v) C(s, v) + C(v, t \land s) - C(v, t)C(v, s)\\
    & - C(t \land v, s \land v) + C(t, v) C(v, s) - C(s \land v, t \land v) + C(v, t)C(s, v).  
\end{align*}
Similarly, if $\widehat{f}^{^{\hspace{0.5mm}S}}_v (t) = \widehat{D}_n (t, v) - \widehat{D}_n (v, t)$, we have 
\begin{align*}
    \sqrt{n} \left \{ \widehat{f}^{^{\hspace{0.5mm}S}}_v (t) - f^{^{S}}_v (t) \right \} \cid \widehat{\mathbb{E}}^{^S}_v (t) := & \hspace{2mm} \mathbb{C}(t, v) - \dot{C}_1 (t, v) \mathbb{C}(t, 1) - \dot{C}_2 (t, v) \mathbb{C}(1, v) \\
    & - \mathbb{C}(v, t) + \dot{C}_1 (v, t) \mathbb{C}(v, 1) + \dot{C}_2 (v, t) \mathbb{C}(1, t)
\end{align*}
for all $t \in [0, 1]$, where $\widehat{\mathbb{E}}^{^S}_v$ is a centered Gaussian random field. One can easily derive its covariance function, but the expression is omitted here in view of its intricate closed form. 
\\
\\
\subsection{\textit{Proof of Proposition~\ref{prop:a2}}}\label{A:prop2}

We first consider the case when $\widehat{f}^{^{\hspace{0.5mm}R}}_v (t) = \widehat{C}_n(t, v) - \widehat{C}_n (1- t, 1 -v) + 1 - t - v$. Notice that 
\begin{align*}
    \sqrt{n} \left \{\widehat{f}^{^{\hspace{0.5mm}R}}_v (t) - f^{^R}_v (t) \right \} &= \sqrt{n} \left \{ \widehat{C}_n(t, v) - \widehat{C}_n (1- t, 1- v) - \big ( C(t,v) - C(1 -t, 1 -v) \big )\right \} \\ 
    & = \sqrt{n} \left \{\widehat{C}_n (t, v) - C(t, v)\right\} - \sqrt{n} \left \{\widehat{C}_n (1 - t, 1- v) - C(1- t, 1 - v)\right\} \\ 
    & = \mathbb{C}_n (t, v) - \mathbb{C}_n (1 - t, 1 - v).
\end{align*}
Since the map
\begin{equation*}  
    \mathcal{T}^R_v: \ell([0, 1]^2) \rightarrow \ell([0, 1]) : f \mapsto \big (t \mapsto f(t, v) - f(1 - t, 1- v)\big)\\
\end{equation*}
is a continuous functional, the Continuous Mapping Theorem and result~\eqref{eq:ec1} imply that $\sqrt{n} \left \{ \widehat{f}^{^{\hspace{0.5mm}R}}_v (t) - f^{^{R}}_v (t) \right \}$ $\cid \mathbb{E}^{^R}_v (t) := \mathbb{C} (t, v) - \mathbb{C}(1 - t, 1- v)$ for all $t \in [0, 1]$, where $\mathbb{E}^{^R}_v$ is a centered Gaussian random field with covariance function given at each $s, t \in [0, 1]$ by 
\begin{align*}
    \Cov{&\mathbb{E}^{^R}_v (t)}{\mathbb{E}^{^R}_v (s)} \\
    = & \hspace{2mm} \Cov{\mathbb{C}(t, v) - \mathbb{C}(1 -t, 1 - v)}{\mathbb{C}(s, v) - \mathbb{C}(1 -s, 1 -v)}\\
    = & \hspace{2mm} C(t \land s, v) - C(t, v) C(s, v) + C(1 - t \lor s, 1 - v) - C(1 -t , 1 - v)C(1 -s, 1 - v)\\
    & - C(t \land (1 - s), v \land (1 - v)) + C(t, v) C(1 - s, 1 -v) \\
    & - C( (1 - t) \land s, v \land (1 - v)) + C(s, v)C(1 - t, 1 - v),
\end{align*}
where for all $a, b \in \R$, $a \lor b = \max \{a, b\}$. 
Similarly, if $\widehat{f}^{^{\hspace{0.5mm}R}}_v (t) = \widehat{D}_n (t, v) - \widehat{D}_n (1 - t, 1 - v) + 1 - t - v$, we have 
\begin{align*}
    \sqrt{n} &\left \{\widehat{f}^{^{\hspace{0.5mm}R}}_v (t) - f^{^{R}}_v (t) \right \} = \mathbb{D}_n (t, v) - \mathbb{D}_n (1 - t, 1 - v)\\
    & \cid \hspace{2mm} \widehat{\mathbb{E}}^{^R}_v (t) := \mathbb{C}(t, v) - \dot{C}_1 (t, v) \mathbb{C}(t, 1) - \dot{C}_2 (t, v) \mathbb{C}(1, v) \\
    & \hspace{2.5cm}- \mathbb{C}(1 - t, 1 -v) + \dot{C}_1 (1 - t, 1 - v) \mathbb{C}(1 -t, 1) + \dot{C}_2 (1 -t , 1 - v) \mathbb{C}(1, 1 - v)
\end{align*}
for all $t \in [0, 1]$, where $\widehat{\mathbb{E}}^{^R}_v$ is a centered Gaussian random field. One can derive its covariance function, but the expression is omitted here in view of its intricate closed form. 

\subsection{\textit{Proof of Proposition~\ref{prop:a3}}}\label{A:prop3}
 The proof of this Proposition is similar to those of Propositions~\ref{prop:a1} and~\ref{prop:a2} but with a different choice of continuous functionals, and we will only present the details of the asymptotic result regarding the test function $\widehat{f}^{^{\hspace{0.5mm}S, 1}}_v$. Firstly, consider the case when $\widehat{f}^{^{\hspace{0.5mm}S, 1}}_v (t) = \widehat{C}_n(t, v) + \widehat{C}_n (t, 1 -v) - t$. Notice that 
\begin{align*}
    \sqrt{n} \left \{\widehat{f}^{^{\hspace{0.5mm}S, 1}}_v (t) - f^{^{S,1}}_v (t) \right \} &= \sqrt{n} \left \{ \widehat{C}_n(t, v) + \widehat{C}_n (t, 1- v) - \big ( C(t,v) + C(t, 1 -v) \big )\right \} \\ 
    & = \sqrt{n} \left \{\widehat{C}_n (t, v) - C(t, v)\right\} + \sqrt{n} \left \{\widehat{C}_n (t, 1- v) - C(t, 1 - v)\right\} \\ 
    & = \mathbb{C}_n (t, v) + \mathbb{C}_n (t, 1 - v).
\end{align*}
Since the map
\begin{equation*}  
    \mathcal{T}^R_v: \ell([0, 1]^2) \rightarrow \ell([0, 1]) : f \mapsto \big (t \mapsto f(t, v) + f(t, 1- v)\big)\\
\end{equation*}
is a continuous functional, the continuous mapping theorem  and result~\eqref{eq:ec1} imply that 
\begin{equation*}
    \sqrt{n} \left \{ \widehat{f}^{^{\hspace{0.5mm}J,1}}_v (t) - f^{^{J,1}}_v (t) \right \}\cid \mathbb{E}^{^{J,1}}_v (t) := \mathbb{C} (t, v) + \mathbb{C}(t, 1- v)~\forall t \in [0, 1],
\end{equation*}
  where $\mathbb{E}^{^{J,1}}_v$ is a centered Gaussian random field with covariance function given at each $s, t \in [0, 1]$ by 
\begin{align*}
    \Cov{&\mathbb{E}^{^{J,1}}_v (t)}{\mathbb{E}^{^{J,1}}_v (s)} \\
    = & \hspace{2mm} \Cov{\mathbb{C}(t, v) + \mathbb{C}(t, 1 - v)}{\mathbb{C}(s, v) + \mathbb{C}(s, 1 -v)}\\
    = & \hspace{2mm} C(t \land s, v) + C(t, v) C(s, v) + C(t \land s, 1 - v) + C(t , 1 - v)C(s, 1 - v)\\
    & + C(t \land s, v \land (1 - v)) + C(t, v) C(s, 1 -v) + C( t \land s, v \land (1 - v)) + C(s, v)C(t, 1 - v),
\end{align*}
where for all $a, b \in \R$, $a \lor b = \max \{a, b\}$. 
Similarly, if $\widehat{f}^{^{\hspace{0.5mm}S,1}}_v (t) = \widehat{D}_n (t, v) + \widehat{D}_n ( t, 1 - v) - t$, we have 
\begin{align*}
    \sqrt{n} &\left \{\widehat{f}^{^{\hspace{0.5mm}R}}_v (t) - f^{^{R}}_v (t) \right \} = \mathbb{D}_n (t, v) + \mathbb{D}_n (t, 1 - v)\\
    & \cid \hspace{2mm} \widehat{\mathbb{E}}^{^{S,1}}_v (t) := \mathbb{C}(t, v) - \dot{C}_1 (t, v) \mathbb{C}(t, 1) - \dot{C}_2 (t, v) \mathbb{C}(1, v) \\
    & \hspace{2.5cm}- \mathbb{C}(1 - t, 1 -v) + \dot{C}_1 (1 - t, 1 - v) \mathbb{C}(1 -t, 1) + \dot{C}_2 (1 -t , 1 - v) \mathbb{C}(1, 1 - v)
\end{align*}
for all $t \in [0, 1]$, where $\widehat{\mathbb{E}}^{^{S,1}}_v$ is a centered Gaussian random field. One can derive its covariance function, but the expression is omitted here in view of its intricate closed form. 
\begin{lem} \label{lem:amean}
	The average of bivariate copulas is also a bivariate copula. 
\end{lem}
\begin{proof}
	Let $n \in \N$ be fixed, and let $C_1, \ldots, C_n$ be $n$ (not necessarily mutually distinct) bivariate copulas. Define 
	\begin{equation*}
	    \mathcal{C}(u, v) = \frac{1}{n} \sum_{i = 1}^n C_i (u, v), \hspace{5mm} \forall (u, v) \in [0, 1]. 
	\end{equation*}
If $u$ or $v$ is zero, say $u = 0$, we have that 
	\begin{equation*}
		\mathcal{C}(0, v) = \frac{1}{n} \sum_{i = 1}^n C_i (0, v) = \frac{1}{n} \sum_{i = 1}^n 0 = 0.
	\end{equation*}
If $u$ or $v$ is $1$, say $u = 1$, one can show that 
	\begin{equation*}
		\mathcal{C}(1, v) = \frac{1}{n} \sum_{i = 1}^n C_i (1, v) = \frac{1}{n} \sum_{i = 1}^n v = v.
	\end{equation*}
For any $(a_1, a_2), (b_1, b_2) \in [0, 1]^2$ with $a_1 \le b_1$ and $a_2 \le b_2$, we have 
	\begin{multline*}
		\sum_{i_1 = 0}^1 \sum_{i_2 = 0}^1  (-1)^{i_1 + i_2} \mathcal{C}(x_{1 i_1}, x_{2 i_2}) = \sum_{i_1 = 0}^1 \sum_{i_2 = 0}^1 (-1)^{i_1 + i_2} \Big ( \frac{1}{n} \sum_{j = 1}^n C_j (x_{1 i_1}, x_{2 i_2}) \Big ) \\ = \frac{1}{n} \sum_{j = 1}^n \Big ( \sum_{i_1 = 0}^1 \sum_{i_2 = 0}^1 (-1)^{i_1 + i_2}  C_j(x_{1 i_1}, x_{2 i_2})\Big) \ge 0,
	\end{multline*}
	where $x_{j0} = a_j$ and $x_{j1} = b_j$ for all $j \in \{1, 2\}$. All of this proves that $\mathcal{C}$ is a copula. Note that the result should be easily extended to any mixture of bivariate copulas.
\end{proof}

\subsection{\textit{Proof of Proposition~\ref{prop:symc}}}\label{A:prop4}
 Assume that $(U,V)$ is a $C$-distributed random vector. Considering the random vector $(V, U)$, we have 
\begin{equation*}
    \mathds{P}(V \le u, U \le v) =   \mathds{P} (U \le v, V \le u) = C(v, u), \hspace{5mm} \forall (u, v) \in [0, 1]. 
\end{equation*}
It follows from Sklar's theorem that the bivariate function $\mathcal{C}$ defined, for all $(u, v) \in [0, 1]$, by $\mathcal{C} (u, v) = C(v, u)$ is a copula. Thus, Proposition~\ref{prop:symc} ensures that $\tilde{C}$ is also a copula. As for the reflection symmetric structure, it is clear by the associativity of addition. 

\subsection{\textit{Proof of Proposition~\ref{prop:rsymc}}}\label{A:prop5}

Proposition~\ref{prop:symc} ensures that $\tilde{C}$ is a bivariate copula. Moreover, one can easily notice
	\begin{align*}
		&\tilde{C}(u, v) - \tilde{C} (1- u, 1- v) \\
		 &=  \frac{1}{2} \Big ( C(u, v) + C(1 - u, 1-v) + u + v - 1 \Big) \\
		 & \hspace{3cm}- \frac{1}{2} \Big ( C(1- u, 1- v) + C(u, v) + (1 - u) + ( 1 - v) - 1 \Big) \\ 
		 &= u + v - 1,  
	\end{align*}
	which tells us that $\tilde{C}$ is radially symmetric. 

\subsection{\textit{Proof of Proposition~\ref{prop:jsymc}}}\label{A:prop6}
To justify that $\tilde{C}$ is a copula, it suffices to show that the bivariate functions $C_1$ and $C_2$ given, for all $(u, v) \in [0, 1]$, by 
\begin{equation*}
    C_1 (u, v) = u - C(u, 1 - v) \text{ \hspace{2mm}and\hspace{2mm} }
    C_2 (u, v) = v - C(1 - u, v)
\end{equation*}
are both copulas. Let $(U, V)$ be a $C$-distributed random vector. Consider the random vector $(U, 1 - V)$. For all $(u, v) \in [0, 1]$, we have 
\begin{multline*}
    \mathds{P}(U \le u, 1 - V \le v) = \mathds{P}(U \le u, 1 - V \le v) \\
    = \mathds{P} (U \le u) - \mathds{P}(U \le u, V \le 1 - v) = u - C(u, 1-v), \hspace{5mm} \forall (u, v) \in [0, 1]. 
\end{multline*}
It follows from Sklar's theorem that $C_1$ is a copula; similarly, $C_2$ can be shown to be the copula of $(1 - u, v)$. Finally, one can check that 
	\begin{align*}
		\tilde{C}&(u, v) + \tilde{C} (u,  1- v)  \\ 
		&= \frac{1}{4} \Big ( C(u, v) + u - C(u, 1 - v) + v - C(1-u, v) + u + v - 1 + C(1- u, 1- v ) \Big)  \\ 
		& \hspace{0.5cm} + \frac{1}{4} \Big ( C(u, 1 - v) + u - C(u, v) + 1 - v - C(1-u, 1 - v) \\
		& \hspace{10cm}+ u + 1 - v - 1 + C(1- u, v ) \Big)  \\ 
		&= \frac{1}{4} (2u + 2v - 1 + 2u - 2v +1 ) = u
	\end{align*} 
	and similarly verify $\tilde{C} (u, v) + \tilde{C}(1- u, v) = v$. With all of this, it can be concluded that $\tilde{C}$ is a jointly symmetric copula. 
\\

\end{document}